\documentclass[a4paper]{article}

\usepackage{graphicx, epsfig}
\usepackage{amssymb, amsthm, amsmath, amsfonts}
\usepackage{bbm}
\usepackage{enumerate}
\usepackage{setspace}
\usepackage[margin=1in]{geometry}
\usepackage{pdflscape}
\usepackage{makeidx}
\usepackage{hyperref}
\usepackage{authblk}
\usepackage[lined,boxed,commentsnumbered]{algorithm2e}
\usepackage{multirow}
\usepackage{subcaption}

\usepackage{amssymb, amsthm, amsmath, amsfonts}
\theoremstyle{definition}
\newtheorem{theorem}{Theorem}[section]
\newtheorem{example}[theorem]{Example}
\newtheorem{corollary}[theorem]{Corollary}
\newtheorem{lemma}[theorem]{Lemma}
\newtheorem{proposition}[theorem]{Proposition}
\newtheorem{definition}[theorem]{Definition}

\newcommand{\rn}{\mathbb{R}^n}
\newcommand{\rd}{\mathbb{R}^d}
\newcommand{\rk}{\mathbb{R}^k}

\newcommand{\rr}{\mathbb{R}}

\renewcommand{\rm}{\mathbb{R}^m}

\newcommand{\E}[2][]{\mathbb{E}_{#1} \left[ #2 \right]}
\newcommand{\conE}[3][]{\E[#1]{#2 | #3}}
\newcommand{\Prob}[1]{\mathbb{P}\left(#1\right)}

\newcommand{\norm}[1][\cdot]{\left\| \kern.05em #1 \kern.05em \right\|}
\newcommand{\var}{\beta\operatorname{-VaR}}
\newcommand{\cvar}{\beta\operatorname{-CVaR}}

\newcommand{\conic}[1]{\operatorname{conic}\left( #1 \right)}
\newcommand{\elliptical}[3]{\operatorname{Elliptical}\left(#1, #2, #3\right)}

\newcommand{\interior}[1]{\operatorname{int}\left(#1\right)}

\newcommand{\minimize}[1][]{\underset{#1}{\operatorname{minimize}}\ }
\newcommand{\maximize}[1][]{\underset{#1}{\operatorname{maximize}}\ }

\newcommand{\aggregion}{\mathcal{A}}
\newcommand{\riskregion}{\mathcal{R}}

\title{Scenario generation for single-period portfolio selection problems with tail risk measures: coping with high dimensions and integer variables}
\author[*]{Jamie Fairbrother}
\author[*]{Amanda Turner}
\author[**]{Stein W. Wallace}
\affil[*]{STOR-i Centre for Doctoral Training, Lancaster University. United Kingdom}
\affil[**]{Department of Business and Management Science, Norwegian School of Economics. Norway}
\graphicspath{{images/}}

\begin{document}
\maketitle

\begin{abstract}
  In this paper we propose a problem-driven scenario generation
  approach to the single-period portfolio selection problem which use
  tail risk measures such as conditional value-at-risk.  Tail risk
  measures are useful for quantifying potential losses in worst cases.
  However, for scenario-based problems these are problematic: because
  the value of a tail risk measure only depends on a small subset of the
  support of the distribution of asset returns, traditional scenario
  based methods, which spread scenarios evenly across the whole support
  of the distribution, yield very unstable solutions unless we use a
  very large number of scenarios. The proposed approach works by
  prioritizing the construction of scenarios in the areas of a
  probability distribution which correspond to the tail losses of
  feasible portfolios.

  The proposed approach can be applied to difficult instances of the
  portfolio selection problem characterized by high-dimensions,
  non-elliptical distributions of asset returns, and the presence of
  integer variables. It is also observed that the methodology works better
  as the feasible set of portfolios becomes more constrained. Based on
  this fact, a heuristic algorithm based on the sample average approximation
  method is proposed. This algorithm works by adding artificial constraints
  to the problem which are gradually tightened, allowing one to telescope
  onto high quality solutions.

\end{abstract}

\section{Introduction}
\label{sec:portfolio-intro}

In the portfolio selection problem one must decide how to invest in a
collection of financial instruments with uncertain returns which in
some way balances portfolio return against risk. There are many ways of modeling
this problem. In the robust optimization setting, the returns
of the portfolio are assumed to fall within some uncertainty set, and
one minimizes the worst-case of loss \cite{bertsimas2006robust}. This
approach is sometimes considered too conservative as
it does not make effective use of available information. 
In stochastic programming, the user uses their knowledge and available data to 
explicitly model asset returns as random vectors, and then optimize
some combination of expected return and risk measure. In between these two
paradigms is distributionally-robust optimization \cite{Huang2010185}
where the returns are modeled by a random vector whose distribution lies in some
uncertainty set, and the worst-case expected loss is minimized. 
This approach is particularly useful when only limited or unreliable data is available to model the uncertainty, but can lead to intractable problems. Besides the
optimization paradigm employed, portfolio selection problems
can be further categorized into single-period problems, where only
one portfolio selection is made, and multiperiod problems where the portfolio
may be rebalanced several times.

The work in this paper applies to the stochastic programming
single-period formulation of the portfolio selection problem. This
approach is popular as it allows one to flexibly model the return of
the distributions, can easily incorporate problem details such as
transactions costs \cite{lobo2007portfolio}, while remaining generally
more tractable than other types of models. We deal in particular with
a difficult variety of this problem type which involve tail risk
measures and potentially integer variables.

In the typical set-up the uncertain returns are modeled by random
variables, the total return of a portfolio is some linear combination
of these, and riskiness is measured by a real-valued function of the
total return which should in some way penalize potential large
losses. This approach was first proposed by Markowitz \cite{Mark52}
who used variance as a risk measure. 

Although the use of variance
has remained popular \cite{lobo2007portfolio} because it leads to tractable
convex programs, its use as a measure of risk is problematic for a few
reasons.  The foremost of these is perhaps that variance penalizes
large profits as well as large losses. As a consequence, in the case
where the returns of financial assets are not normally distributed,
using the variance can lead to potentially bad decisions; for instance, a
portfolio can be chosen in favor of one which always has higher
returns (see \cite{Young98} for an example of this).  This particular
issue can be overcome by using a ``downside'' risk measure, that is
one which only depends on losses greater than the mean, or some other
specified threshold, for example the semi-variance 
\cite[Chapter~9]{Markowitz}, mean regret \cite{DeRo99}, or value-at-risk
\cite{Jo96}. More recently, much research has been given
to coherent risk measures, a concept introduced in \cite{Artzner98}.
These are risk measures which have sensible properties such as subadditivity,
which in particular ensures that a risk measure incentivizes diversification
of a portfolio. Using a coherent risk measure in a portfolio selection problem should avoid flawed
decisions, such as the one cited in the case of variance.

In this work, we are interested in portfolio selection problems
involving \emph{tail risk measures}.  These can be thought of as risk
measures which only depend on the upper tail of a distribution above
some specified quantile. A canonical example of a tail risk measure
is the value-at-risk (VaR) \cite{Jo96}. The
$\beta$-VaR is defined to be the $\beta$-quantile of a random
variable. In portfolio selection problems this has the appealing
interpretation as the amount of capital required to cover up to
$\beta\times 100 \%$ of potential losses. Thus, tail risk measures, in particular
those which dominate the $\var$, are useful as they can give us some idea of
the amount capital at risk in the worst  $(1-\beta) \times 100 \%$ of potential
losses. Like variance, the value-at-risk is also problematic as it
is not a coherent measure of risk. Specifically, it is not
subadditive (see \cite{tasche2002expected} for example). Moreover, $\var$ leads to difficult and intractable problems
when used in an optimization context. The conditional value-at-risk (CVaR),
sometimes referred to as the expected shortfall, is another tail risk measure and can be roughly thought as
the conditional expectation of a random variable above the $\beta$-VaR.
It is both coherent \cite{acerbi2002coherence}, and more tractable in an optimization setting \cite{Rockafellar00}.

However, the use of risk measures, even coherent ones such
as $\cvar$, is still problematic in portfolio selection
problems where the asset returns are modeled with continuous
probability distributions. This is because the evaluation
of many risk measures for arbitrary continuously distributed returns
would involve the evaluation of a multidimensional integral, and
this becomes computationally infeasible when
our problems involve many assets. On the other hand, the evaluation
of such an integral reduces to a summation if the returns have a discrete
distribution. 

Scenario generation is the
construction of a finite discrete distribution to be used in a
stochastic optimization problem. This may involve fitting a parametric
model to asset returns
and then discretizing this distribution, or directly modeling them
with a discrete distribution, for example via moment-matching
\cite{HoylandEA00}.  In either case, standard scenario generation
methods struggle to adequately represent the uncertainty in problems
using tail risk measures. This is because the value of a tail risk
measure, by definition, only depends on a small subset of the support
of a random variable, and typical scenario generation methods will
spread their scenarios evenly across the whole support of the
distribution. This means that the region on which the value of the
tail risk depends, is represented by relatively few scenarios. Hence,
unless there is a very large number of scenarios, the value of of
tail risk measure is very unstable (see \cite{KautEA03} for example).

The natural remedy to this problem is to represent the regions of the distribution
on which the tail risk measure depends
with more scenarios. Intuition would tell us that these
correspond to the ``tails'' of the distribution. However, for a multivariate
distribution there is no canonical definition of the tails. If by tails,
we simply mean the region where at least one of the components exceeds
a large value, then the probability of this region quickly converges to one with the problem dimension,
and thus prioritizing scenarios in this region will be of little benefit.
Finding the relevant tails of the distribution is a non-trivial problem.

In our previous paper \cite{Fairbrother15a} we addressed the problem of
scenario generation for stochastic programs with an arbitrary loss function
functions which use tail risk measures,
and for this we defined the concept of a $\beta$-risk region. 
In portfolio selection, to each valid portfolio there is a distribution
of losses (or returns). The $\beta$-risk region consists of all
potential asset returns which lead to a loss in the $\beta$-tail for some
portfolio. We have shown that under mild conditions the value of a tail risk measure in effect only depends
on the distribution of returns in the risk region. Although characterizing this region in a convenient way is
generally not possible, we have been able to do this for the portfolio selection
problem when the asset returns are elliptically distributed. We have also
proposed a sampling approach to scenario generation using these risk regions
which prioritizes the generation of scenarios in the risk region.  We
demonstrated for simple examples that this methodology can produce
scenario sets which yield better and more stable solutions than does
basic sampling. 

In this paper we address issues related to the application of this
methodology to realistic portfolio selection problems. The first
major contribution of this work is the application of the methodology to
problems where the asset returns have non-elliptical distributions. In this
case the distribution of returns for a portfolio will in general
not have a convenient closed form, and so it is necessary to represent
the asset returns with a scenario set. In order to apply the methodology, we approximate the risk region for a non-elliptical distribution with the risk region
of an elliptical distribution. Moreover, we demonstrate
that the methodology is effective on difficult problems which have
high dimensions and integer variables.

In the paper \cite{Fairbrother15a} it was shown that our methodology
is more effective as the problem becomes more constrained. The second
major contribution of this work is the proposal of an heuristic
algorithm based on the stochastic average approximation (SAA) method
\cite{KleywegtEA01} which exploits this fact. This algorithm works by
adding artificial constraints to the problem which are gradually
tightened, allowing one to telescope onto high quality solutions.
This algorithm is presented in a general way and could be potentially
used on problems other than portfolio selection.

The paper is organized as follows. In Section~\ref{sec:portfolio-recap}
we formally define risk regions for general stochastic programs, and
present some new results on the use of approximate risk regions. 
In Section~\ref{sec:risk-regi-portf} we define the portfolio selection 
problem and recall some results on the use of the risk region methodology
for this problem. We also provide some new technical results related
exploiting risk regions of elliptical distributions.
In Section~\ref{sec:portfolio-scengen} we describe how risk regions
are exploited for the purpose of scenario generation, and present
a heuristic based on the SAA method based which uses of artificial
constraints. In Section~\ref{sec:portfolio-probnonrisk} we make some
empirical observations on how the probability of the non-risk region,
a quantity which determines the effectiveness of the methodology,
varies with the type of distribution of asset returns. 
In Section~\ref{sec:portfolio-numtests} we present results
for a broad range of numerical tests on the effectiveness of our
sampling and reduction algorithm using distributions constructed
from real asset return data. In Section~\ref{sec:portfolio-case}
we demonstrate the performance of the proposed heuristic on a
difficult case study problem. Finally, in Section~\ref{sec:portfolio-conclusions}
we make some concluding remarks.

\section{Risk regions for general stochastic programs}
\label{sec:portfolio-recap}

In this section we formally define the concept of risk region and present
results related to these. This theory does not just apply to portfolio selection
problems but more generally to stochastic programs with a tail risk measure.
In Section~\ref{sec:portfolio-recap-tail-risk} we recall the basic definitions
and fundamental results for risk regions which appeared in our previous
paper \cite{Fairbrother15a}. In Section~\ref{sec:appr-risk-regi} we present
some new results related to the approximation of risk regions.

\subsection{Tail risk measures and risk regions}
\label{sec:portfolio-recap-tail-risk}

A risk measure is
a function of a real-valued random variable representing
a loss. For $0 < \beta \leq 1$, a $\beta$-tail risk measure can be thought of as a
function of a random variable which depends only on the upper
$(1 - \beta)$-tail of the distribution. The precise definition uses the
\emph{generalized inverse distribution function} or (lower) \emph{quantile function}.

\begin{definition}[Quantile function and $\beta$-tail risk measure]
  Suppose $Z$ is a random variable with distribution function $F_Z$. Then 
  the generalized inverse distribution function, or \emph{quantile function}
  is defined as follows:
  \begin{align*}
    F^{-1}_Z : (0, 1] &\rightarrow \rr\cup\{\infty\}\\
    \beta &\mapsto \inf\{ z\in\rr: F_{Z}(x) \geq \beta \}
  \end{align*}
  Now a $\beta$-tail risk measure is any function of a random variable, $\rho_\beta(Z)$, which
  depends only on the quantile function of a random variable above $\beta$.
\end{definition}

\begin{example}[Value at risk (VaR)]
  Let $Z$ be a random variable, and $0< \beta < 1$. Then, the $\beta-$VaR for $Z$ is defined to be the $\beta$-quantile of $Z$:
  \begin{equation*}
    \var(Z) := F_Z^{-1}(\beta)
  \end{equation*}.
\end{example}

\begin{example}[Conditional value at risk (CVaR)]
  Let $Z$ be a random variable, and $0 < \beta < 1$. Then, the
  $\cvar$ can be thought roughly as the conditional
  expectation of a random variable above its $\beta$-quantile.
  The following alternative characterization of $\cvar$
  \cite{acerbi2002coherence} shows directly that it is
  a $\beta$-tail risk measure.
  \begin{equation*}
    \cvar(Z) = \int_{\beta}^1 F^{-1}_Z(u)\ du
  \end{equation*}
\end{example}

The observation that we exploit for this work is that very
different random variables will have the same $\beta$-tail risk
measure as long as their $\beta$-tails are the same. 

In the optimization context we suppose that the loss depends on some
decision $x\in \mathcal{X} \subseteq \rk$ and the outcome of some
latent random vector $Y$ with support $\mathcal{Y} \subseteq \rd$,
defined on a probability space $(\Omega, \mathcal{F}, \mathbb{P})$,
and which is independent of $x$. That is, we suppose our loss is
determined by some function, $f: \mathcal{X}\times\rd \rightarrow
\rr$, which we refer to as the \emph{loss function}.  For a given
decision $x\in\mathcal{X}$, the random variable associated with the
loss is thus $f(x, Y)$.

To avoid repeated use of cumbersome notation we introduce the following
short-hand for distribution and quantile functions:
 \begin{align*}
   F_x(z) &:= F_{f(x,Y)}(z) = \Prob{f(x,Y)\leq z},\\
   F_x^{-1}(\beta) &:= F_{f(x,Y)}^{-1}(\beta) = \inf\{z\in\rr:\ F_x(z) \geq \beta\}.
 \end{align*}

In \cite{Fairbrother15a} we introduced the concept of a
risk region for a stochastic program using a tail-risk measure.
We define this now for a general stochastic program.

\begin{definition}[Risk region]
  \label{def:risk-region}
The $\beta$-risk region associated associated with the random vector
$Y$ and the feasible region $\mathcal{X}\subseteq\rd$ is as follows:
\begin{equation}
  \label{eq:portfolio-riskregion-1}
  \mathcal{R}_{Y, \mathcal{X}}(\beta) := \bigcup_{x\in\mathcal{X}}\{y\in\rd :\ f(x,y) \geq F^{-1}_{x}\left(\beta\right)\}.
\end{equation}
\end{definition}

The risk region consists precisely of those outcomes of $Y$ which have
a loss in the $\beta$-tail of the loss distribution for \emph{some}
feasible decision. We refer to the complement of the risk
region as the non-risk region and this consists of outcomes which
never lead to a loss in the $\beta$-tail; it can be written as follows:
\begin{equation}
  \label{eq:portfolio-non-risk-region}
  \riskregion_{Y,\mathcal{X}}(\beta)^{c} = \bigcap_{x\in\mathcal{X}} \{y \in \rd : f(x,y) < F_{x}^{-1}(\beta)\}.
\end{equation}

The following theorem was proved in \cite{Fairbrother15a} 
and states that under mild conditions the
value of a tail risk measure is completely determined by the the
distribution of the random vector $Y$ in the risk region. That is, the
values of the tail risk measure of any two random vectors with
identical distributions in the risk region will be the same for all
feasible decisions. The technical condition in
\eqref{eq:portfolio-agg-condition}, which we call the
aggregation condition for reasons explained below, 
precludes certain degenerate cases. In essence, this condition ensures
that there is enough mass in the set to ensure that the
$\beta$-quantile does not depend on the probability distribution
outside of it.

\begin{theorem}
  \label{thr:portfolio-equiv-rv}
  Let $\riskregion \supseteq \riskregion_{Y,\mathcal{X}}(\beta)$ be such that for all $x\in\mathcal{X}$ the following condition holds:
  \begin{equation}
    \label{eq:portfolio-agg-condition}
    \Prob{Y \in \{y: z' < f(x,y) \leq F_{x}^{-1}\left(\beta\right)\}\cap\riskregion} > 0 \qquad \forall\ z' < F^{-1}_{x}\left(\beta\right).
  \end{equation}
  If $\tilde{Y}$ is a random vector for which the following holds:
  \begin{equation}
    \label{eq:portfolio-agg-condition-2}
    \Prob{Y \in \mathcal{A}} = \Prob{\tilde{Y} \in \mathcal{A}}\qquad\text{for any } \mathcal{A}\subseteq \riskregion,
  \end{equation}
  then $\rho_\beta\left(f(x,Y)\right) = \rho_\beta\left(f(x,\tilde{Y})\right)$ for all $x\in \mathcal{X}$, for any $\beta$-tail risk measure $\rho_\beta$.
\end{theorem}

With regards to scenario generation, this theorem says that any
scenarios in the non-risk region can be aggregated into a single point, reducing
the size of the problem, without affecting the value of the
tail risk measure. This motivates the term \emph{aggregation condition}
for \eqref{eq:portfolio-agg-condition}. 
The transformed random vector where all mass in a region
has been concentrated into its conditional expectation plays a special role
in this work. We call this the \emph{aggregated random vector}.

\begin{definition}[Aggregated Random Vector]
  \label{def:tailrisk-agg-rv}
  For some set $\riskregion\supseteq\riskregion_{Y,\mathcal{X}}$ the
  \emph{aggregated random vector} is defined as follows:
\begin{equation*}
  \psi_{\riskregion}(Y) := \begin{cases} Y &\text{if } Y \in \riskregion,\\
    \E{\ Y | Y\in\riskregion^c\ } & \text{otherwise.} \end{cases}
\end{equation*}
\end{definition}

The conditional expectation $\E{\ Y | Y\in\riskregion^c\ }$ is guaranteed to
fall in the non-risk region if, for example, the loss function
is convex \cite[Proposition~3]{Fairbrother15a}.

\subsection{Approximation of Risk Regions}
\label{sec:appr-risk-regi}

The methodology proposed in this paper requires a characterization
of a risk region which allows one to easily test membership. A
convenient characterization for the exact risk region
as defined in \eqref{eq:portfolio-riskregion-1} is in general
not possible as this set is determined by the loss function,
distribution of the random vector and the problem constraints.
Even for the portfolio selection which has a simple loss function,
we cannot find a convenient form for arbitrary distributions
of asset returns.

Recall that Theorem~\ref{thr:portfolio-equiv-rv}
applies to any set containing the risk region.
Therefore, one way to circumvent the problem of
finding the exact risk region would be to use
a \emph{conservative risk region}, that is, a set
which contains the exact risk region. This approach
is particularly useful for problems which have
constraints which cannot be easily taken into
account, such as constraints involving
integer variables in the portfolio selection problem.
By Definition~\ref{def:risk-region}, if $\mathcal{X}\subseteq\mathcal{X}'$
then $\riskregion_{Y,\mathcal{X}} \subseteq \riskregion_{Y,\mathcal{X}'}$. Therefore, ignoring
some constraints will yield a risk region which is conservative.

In the case where one cannot construct a risk region
for a given loss function or distribution, it may be
difficult to find a conservative risk region. Moreover, it could
be the case that a conservative risk region may be too conservative
to be of any use. Instead one might try to use an
\emph{approximate risk region}. For the portfolio selection
problem we handle distributions for which we cannot conveniently characterize
the exact risk region by using the risk region of a surrogate distribution
which is similar to the true distribution.

Denote by $\riskregion\subset\rd$ an approximate risk region. 
When using an approximate risk region
the value of the tail risk measure may be distorted
for a decision $x\in\mathcal{X}$ if $\riskregion$ does not contain
all outcomes in the $\beta$-tail, that is, unless the following condition
holds:
\begin{equation}
  \label{eq:approx-risk-risk}
  \sup_{y\in\riskregion^{c}} f(x,y) \leq F^{-1}_{x}(\beta).
\end{equation}
If \eqref{eq:approx-risk-risk} holds then we say that
that the approximate risk region $\riskregion$ is valid for decision $x$.

We show in this section that if the approximate risk region is not valid
for a particular decision, then, under mild assumptions, the values of
$\var$ and $\cvar$ tail risk measures are distorted downwards. In
Section~\ref{sec:risk-regi-portf}, we will exploit this observation to
show that for the problems in which we are interested, if there is no
distortion of the value of the tail risk measure at the optimal
solution, then this solution is also optimal with respect to the true
problem.

For the results in this section we employ the following notation:
$\hat{F}_{x}$ and $\hat{F}^{-1}_{x}$ denote respectively the
distribution and quantile functions of $\psi_{\riskregion}(Y)$. We
require the following conditions:

\begin{enumerate}[(A)]
\item $z\mapsto F_{x}(z)$ is continuous for all $x\in\mathcal{X}$
\item $\conE{Y}{Y\in\riskregion^{c}} \in \riskregion_{Y,\mathcal{X}}^{c}$
\end{enumerate}

Assumption (B) requires that conditional expectation of $Y$ in the
complement of approximate risk region belongs to the exact non-risk
region.  This means that the loss at the aggregated point will have a loss
below the $\beta$-quantile of $f(x,Y)$ for all feasible decisions $x\in\mathcal{X}$. 
Before stating and proving the key result, we require the following lemma.

\begin{lemma}
  \label{lem:cont-cdf}
  Under assumptions (A) and (B), the approximate distribution function $\hat{F}_{x}$ is continuous for all $x\in\mathcal{X}$ at $z$ for ${z > f\left(x,\conE{Y}{Y\in\riskregion}\right)}$.
\end{lemma}

\begin{proof}
  Fix $x\in\mathcal{X}$ and $z >
f\left(x,\conE{Y}{Y\in\riskregion}\right)$, and without loss of
generality assume that ${f\left(x,\conE{Y}{Y\in\riskregion}\right)< z' < z}$. Now,
  
  \begin{align*}
    \hat{F}_{x}( z ) - \hat{F}_{x}(z') &= \Prob{f(x, \psi_{\riskregion}(Y)) < z} - \Prob{f(x, \psi_{\riskregion}(Y)) < z'} \\
    & = \left(\Prob{\{Y\in\riskregion\}\cap \{f(x,Y) \leq z\}} + \Prob{Y\in\riskregion^{c}}\right) - \left(\Prob{\{Y\in\riskregion\}\cap \{f(x,Y) \leq z'} + \Prob{Y\in\riskregion^{c}}\right)\\
    &= \Prob{\{Y\in\riskregion\} \cap \{z'< f(x,Y) \leq z\}}\\
    &\leq F_{x}(z) - F_{x}(z') \rightarrow 0 \text{ as } z' \rightarrow z \qquad \text{by assumption (A).}
  \end{align*}
\end{proof}

The key result states that $\beta$-quantile (or $\var$) and $\cvar$
for the aggregated random vector cannot increase when using
an approximate risk region under the above assumptions. The implications
of this result on the portfolio selection problem are made clear
in Section~\ref{sec:probl-stat-appl}.

\begin{proposition}
  \label{prop:low-var}
  Under assumptions (A) and (B), we have
  \begin{itemize}
  \item $\hat{F}^{-1}_{x}(\beta) \leq F^{-1}_{x}(\beta)$
  \item $\cvar\left( f(x, \psi_{\riskregion}(Y)) \right) \leq \cvar\left( f(x, Y) \right)$
  \end{itemize}
  with equality if $\riskregion$ is valid for $x\in\mathcal{X}$ (in the sense of 
  \eqref{eq:approx-risk-risk}) and the aggregation condition holds.
\end{proposition}

\begin{proof}
  \begin{align*}
    \Prob{f(x,\psi_{\riskregion}(Y)) \leq F^{-1}_{x}(\beta)} &= \Prob{\{Y\in\riskregion\}\cap\{f(x,Y) \leq F^{-1}_{x}(\beta)\}} + \Prob{Y\in\riskregion^{c}}\\
    &= \underbrace{\Prob{f(x,Y) \leq F^{-1}(\beta)}}_{=\beta\ \text{by assumption (A)}} + \Prob{\{Y\in\riskregion^{c}\} \cap \{f(x,Y) > F^{-1}_{x}(\beta)\}}\\
     &\geq \beta.
  \end{align*}
  Hence, $\hat{F}_{x}^{-1}(\beta) \leq F^{-1}_{x}(\beta)$.

  For the $\cvar$, recall that for a random variable $Z$ this can be written as follows \cite{acerbi2002coherence}:
  \begin{equation*}
    \cvar(Z) = \frac{1}{1-\beta}\left(\E{Z\mathbbm{1}_{Z\geq F_{Z}^{-1}(\beta)}} - F_{Z}^{-1}(\beta)\left(\beta - \Prob{Z < F^{-1}_{Z}(\beta)}\right)\right)
  \end{equation*}
  where $\mathbbm{1}_{A}$ denotes the indicator function of event $A$.
  Since $F_{x}$ is continuous we can write:
  \begin{equation*}
    \cvar(f(x,Y)) = \frac{1}{1-\beta}\E{f(x,Y)\mathbbm{1}_{f(x,Y)\geq F_{x}^{-1}(\beta)}}.
  \end{equation*}

  On the other hand, $\hat{F}_{x}$ could have a discontinuity at $\hat{F}^{-1}_{x}(\beta)$ if $\hat{F}^{-1}_{x}(\beta) = f\left(x, \conE{Y}{Y\in\riskregion}\right)$.
  We therefore consider two cases:
  \begin{enumerate}
  \item $\hat{F}_{x}^{-1}(\beta) > f\left(x, \conE{Y}{Y\in\riskregion}\right)$
  \item $\hat{F}_{x}^{-1}(\beta) = f\left(x, \conE{Y}{Y\in\riskregion}\right)$
  \end{enumerate}

  In the first case, $\hat{F}_{x}$ in continuous at $\hat{F}_{x}^{-1}(\beta)$ by Lemma~\ref{lem:cont-cdf} so we can write:
  \begin{align*}
    \cvar(f(x,\psi_{\riskregion}(Y))) &= \frac{1}{1-\beta}\E{f(x,\psi_{\riskregion}(Y))\mathbbm{1}_{f(x,\psi_{\riskregion}(Y))\geq F_{x}^{-1}(\beta)}}\\
    &= \int_{\riskregion\cap\{y: \hat{F}^{-1}_{x}(\beta) \leq f(x,y) < F^{-1}_{x}(\beta)\} } f(x,y)\ d\Prob{y} + \int_{\riskregion\cap\{y:f(x,y) \geq F_{x}^{-1}(\beta)\}} f(x,y)\ d\Prob{y}
  \end{align*}
  Therefore,
  \begin{align*}
    \cvar(f(x,Y)) - \cvar(f(x,\psi_{\riskregion}(Y))) & = \frac{1}{1-\beta}\bigg( \int_{\riskregion^{c}\cap\{y:f(x,y) \geq F_{x}^{-1}(\beta)\}} f(x,y)\ d\Prob{y} - \\
  & \qquad  \int_{\riskregion\cap\{y: \hat{F}^{-1}_{x}(\beta) < f(x,y) < F^{-1}_{x}(\beta)\} } f(x,y)\ d\Prob{y} \bigg).
  \end{align*}
  Note that the integrand of the first term is greater than that of
  the second term over the respective domain of integration. Therefore,
  to show that the above quantity is non-negative, it is enough to
  show that the domain of integration of the first term has the same probability
  as the second. To show this, first not that:
  \begin{align*}
    &\Prob{f(x,Y) \leq F^{-1}_{x}(\beta)} = \beta,\\
    &\Prob{f(x,\psi_{\riskregion}(Y))\leq \hat{F}^{-1}_{x}(\beta)} = \Prob{Y\in\riskregion^{c}} + \Prob{\riskregion\cap\{f(x,Y) \leq \hat{F}_{x}^{-1}(\beta)} = \beta.
  \end{align*}
  Therefore,
  \begin{equation*}
    \Prob{f(x,Y) \leq F^{-1}_{x}(\beta)} = \Prob{Y\in\riskregion^{c}} + \Prob{\riskregion\cap\{f(x,Y) \leq \hat{F}_{x}^{-1}(\beta)},
  \end{equation*}
  rearranging which gives
  \begin{equation*}
    \Prob{\riskregion\cap\{\hat{F}^{-1}_{x}(\beta) < f(x,Y) \leq F^{-1}_{x}(\beta)\}} = \Prob{\riskregion^{c}\cap\{f(x,Y) > F_{x}^{-1}(\beta)},
  \end{equation*}
  as required.

  In the second case, $\hat{F}_{x}$ has a discontinuity at $\hat{F}_{x}^{-1}(\beta)$, and so the $\cvar$ is written as follows:
  \begin{align*}
    \cvar\left(f(x,\psi_{\riskregion}(Y))\right) &= \frac{1}{1-\beta}\Bigg(\hat{F}_{x}^{-1}(\beta)\Prob{Y\in\riskregion^{c}} + \int_{\riskregion\cap\{y: \hat{F}^{-1}_{x}(\beta) \leq f(x,y)\} } f(x,y)\ d\Prob{y} - \\
    & \qquad\qquad \hat{F}_{x}^{-1}(\beta)\left(\beta - \Prob{\riskregion\cap\{f(x,Y) \leq \hat{F}_{x}^{-1}(\beta)}\right)\Bigg)\\
  \end{align*}
Noting that,
\begin{align*}
  & \{f(x,Y) \geq F^{-1}_{x}(\beta)\} = \left(\{Y\in\riskregion\}\cap\{f(x,Y) \geq F^{-1}_{x}(\beta)\}\right) \bigcup \left(\{Y\in\riskregion^{c}\}\cap\{f(x,Y) \geq F^{-1}_{x}(\beta)\}\right) \qquad\text{ and}\\
  & \left(\{Y\in\riskregion\}\cap\{f(x,Y) > \hat{F}^{-1}_{x}(\beta)\}\right) \setminus \left(\{Y\in\riskregion\} \cap \{f(x,Y) \geq F^{-1}_{x}(\beta)\}\right) = \{Y\in\riskregion\}\cap\{\hat{F}^{-1}_{x}(\beta) < f(x,Y) < F^{-1}_{x}(\beta)\},
\end{align*}
we can write $\cvar(f(x,Y)) - \cvar(f(x,\psi_{\riskregion}(Y)))$ as:
\begin{align*}
 \frac{1}{1-\beta}\Bigg(&\int_{\riskregion^{c}\cap\{y:f(x,y) \geq F_{x}^{-1}(\beta)\}} f(x,y)\ d\Prob{y} - \int_{\riskregion\cap\{y: \hat{F}^{-1}_{x}(\beta) < f(x,y) < F^{-1}_{x}(\beta)\} } f(x,y)\ d\Prob{y}  \\
  & \qquad - \hat{F}^{-1}_{x}(\beta)\left(\Prob{Y\in\riskregion^{c}} + \Prob{\riskregion\cap\{f(x,Y) < \hat{F}_{x}^{-1}(\beta)\}} - \beta\right)\Bigg) \\
  \geq & \frac{1}{1-\beta}\Bigg( F_{x}^{-1}(\beta)\left( \Prob{\{Y\in\riskregion^{c}\}\cap\{\hat{F}^{-1}_{x}(\beta) < f(x,Y) < F^{-1}_{x}(\beta)\} } - \Prob{\{Y\in\riskregion\}\cap\{\hat{F}^{-1}_{x}(\beta) < f(x,Y) < F^{-1}_{x}(\beta)\}}\right) \\
  & \qquad - \hat{F}_{x}^{-1}(\beta)\left(\Prob{Y\in\riskregion^{c}} + \Prob{\riskregion\cap\{f(x,Y) < \hat{F}_{x}^{-1}(\beta)}\right)\Bigg).
\end{align*}
Finally, manipulation of the probabilities above yields:
\begin{align*}
  \frac{1}{1-\beta}&\Bigg(F_{x}^{-1}(\beta)\left(\Prob{Y\in\riskregion^{c}} + \Prob{\riskregion\cap\{f(x,Y) < \hat{F}_{x}^{-1}(\beta)} - \beta\right) -  \\
  & \qquad \hat{F}_{x}^{-1}(\beta)\left(\Prob{Y\in\riskregion^{c}} + \Prob{\riskregion\cap\{f(x,Y) < \hat{F}_{x}^{-1}(\beta)} -\beta\right)\Bigg) \geq 0,
\end{align*}
  since $\hat{F}^{-1}_{x}(\beta)\leq F^{-1}_{x}(\beta)$, as required.
  
  The fact that the inequalities hold with equality if $\riskregion$ is
  valid for decision $x$ and the aggregation condition holds follows directly from
  Theorem~\ref{thr:portfolio-equiv-rv} for the special case $\mathcal{X}=\{x\}$.
\end{proof}

\section{Risk Regions for Portfolio Selection}
\label{sec:risk-regi-portf}

In this section we present results relating to risk regions
for the portfolio selection problem. 
In Section~\ref{sec:portfolio-recap-portfolio-risk} we define the problem
and present general results from \cite{Fairbrother15a} related to the risk region
for this problem. The remaining two subsections
deal with risk regions for elliptical distributions since these are used
as approximate risk regions. Specifically, in Section~\ref{sec:portfolio-recap-portfolio-risk} we formally define elliptical distributions and give
a convenient characterization of their corresponding risk regions, and in Section~\ref{sec:portfolio-conic} we present some new results related to testing membership to a risk region.

\subsection{Problem statement and application of risk regions}
\label{sec:probl-stat-appl}

We use the following basic set-up: we have
a set of financial assets indexed by $i=1,\ldots,d$, by $x_i$ we
denote how much we invest in asset $i$, and by $Y_i$ we denote the
random future return of asset $i$. The portfolio return associated to a
particular investment decision $x=(x_1,\ldots,x_d)$ and return $Y =
(Y_1,\ldots,Y_d)$ is $x^TY = \sum_{i=1}^{d} x_iY_i$. The loss function associated
to an investment decision is thus $f(x,Y) = -x^TY$, and so for a given $\beta$-tail
risk measure $\rho_\beta$ we would like an investment with small tail risk
$\rho_\beta(-x^T Y)$. The aim of a portfolio selection problem is to choose a decision
which balances choosing a portfolio with high expected portfolio return against
choosing one with small risk. This typically corresponds to solving
a problem of one of the following forms:

\begin{align}
  \text{(i)}\qquad\minimize[x\in\mathcal{X}] &\rho_\beta(-x^TY)\label{eq:p1}\tag{P1}\\
  &\text{subject to} \ \E{x^TY} \geq t,\nonumber\\
  \text{(ii)}\qquad\maximize[x\in\mathcal{X}] & \E{x^TY}\label{eq:p2}\tag{P2}\\
  &\text{subject to} \ \rho_\beta(-x^TY) \leq s,\nonumber\\
  \text{(iii)}\qquad\minimize[x\in\mathcal{X}] &\lambda \rho_\beta(-x^TY) + (1-\lambda)\E{-x^TY},\label{eq:p3}\tag{P3}
\end{align}
where $0\leq \lambda \leq 1$ and $\mathcal{X} \subset \rd$ represents the set of feasible portfolios. This feasibility region will typically encompass a constraint which specifies the amount of capital
to be invested, and may include others which, for example the
exclusion of short-selling, or a limit on the amount that can be
invested in certain industries.

In the case of the portfolio selection problem, the risk region is
$\mathcal{R}_{Y, \mathcal{X}}(\beta) :=
\bigcup_{x\in\mathcal{X}}\{y\in\rd :\ -x^{T}y \geq
F^{-1}_{x}\left(\beta\right)\}$, that is, it is the union over all
feasible portfolios, of the half spaces of points with returns above
the $\beta$-quantile.  We can find this region by brute force, and
this is illustrated for a hypothetical discrete random vector on the
left-hand side of Figure \ref{fig:full_agg}. Also illustrated in this
figure is the set of returns where all the mass in the non-risk region
has been aggregated into the conditional expectation of the random
vector in the non-risk region, that is, the aggregated random
vector. The figure also demonstrates, as implied by
Theorem~\ref{thr:portfolio-equiv-rv}, that the $\beta$-quantile lines
do not change after aggregation.

\begin{figure}[h]
  \centering
  \includegraphics[width=\textwidth]{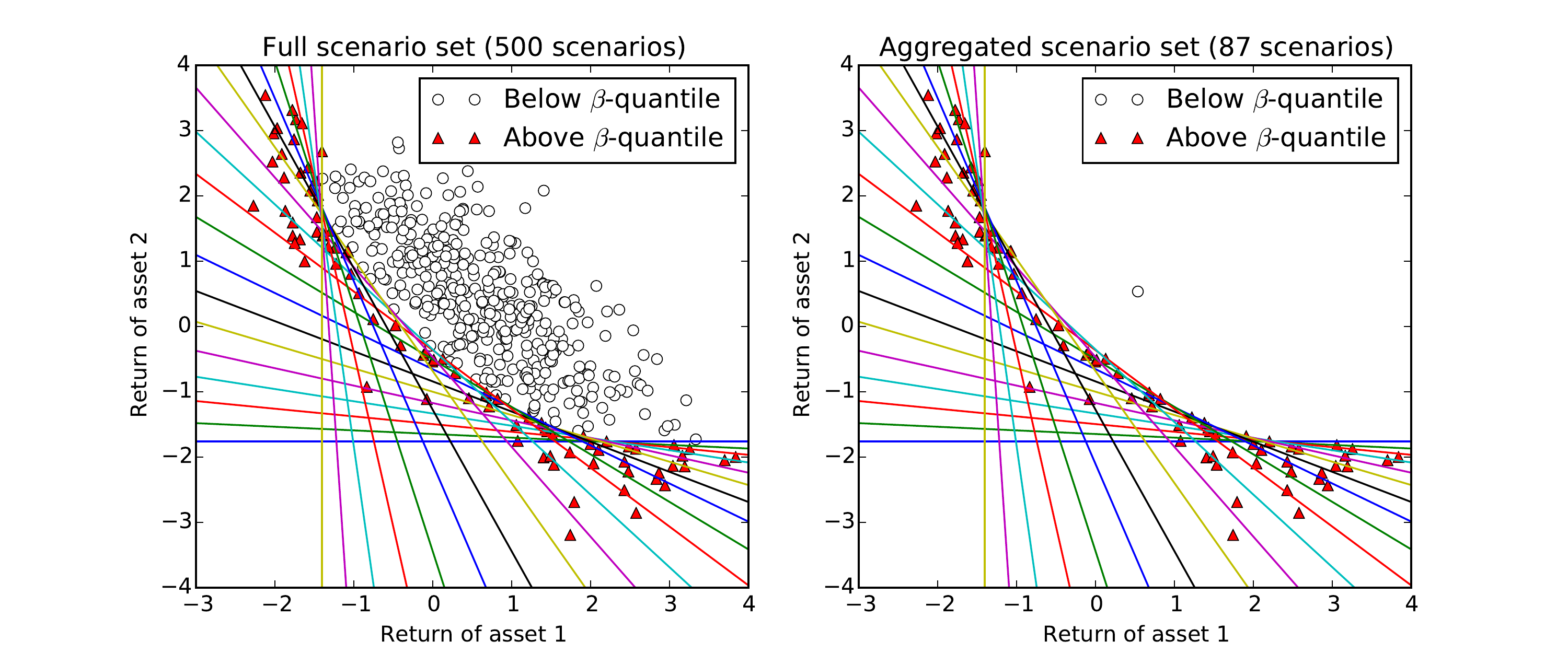}
  \caption{Return points of two assets with loss below the $\beta$-quantile for all non-negative portfolios (left) and aggregated scenario set (right)}
  \label{fig:full_agg}
\end{figure}

Assuming $\E{Y|Y\in\riskregion^{c}}\in\riskregion_{Y,\mathcal{X}}^{c}$, as well as preserving the value 
of a tail risk measure, the aggregated random vector
has the additional property of preserving the overall expected return
of the original random vector. The following corollary taken from
\cite{Fairbrother15a} summarizes this result and provides sufficient
conditions so that \eqref{eq:portfolio-agg-condition} holds.

\begin{corollary}
  \label{cor:portfolio-agg-rv}
  Suppose $\riskregion_{Y,\mathcal{X}}(\beta) \subseteq \riskregion \subset \rd$, $Y$ is a continuous random vector with support $\mathcal{Y} =
  \rd$, and $\mathcal{X}$ contains at least two linearly independent elements. Then $Y$ satisfies \eqref{eq:portfolio-agg-condition}. In addition, if $\riskregion^{c}$ is convex then $\tilde{Y}=\psi_{\riskregion}(Y)$ satisfies condition \eqref{eq:portfolio-agg-condition-2} and so for all $x\in\mathcal{X}$
  we have:
  \begin{align*}
    \rho_\beta\left(-x^{T}Y\right) &= \rho_\beta\left(-x^{T}\tilde{Y}\right),\\
    \E{x^{T}Y} &= \E{x^{T}\tilde{Y}}.
  \end{align*}
\end{corollary}

In Section~\ref{sec:appr-risk-regi} we showed that under mild
conditions when using an approximate risk region, a misspecification
for a particular decision will decrease the value of the $\var$ and
$\cvar$.  Building on this result, the following corollary
gives a condition under which the optimal solution yielded by
using an approximate risk region is also optimal for the true problem.

\begin{corollary}
  \label{cor:approx-risk-region}
  Under assumptions (A) and (B), for the problems \eqref{eq:p1},
  \eqref{eq:p2} and \eqref{eq:p3}, suppose that replacing the solution
  the random vector $Y$ with the aggregated random vector
  $\psi_{\riskregion}(Y)$ for an approximate risk region $\riskregion$
  yields an optimal solution $\hat{x}$. Then, if $\riskregion$ is valid
  for $\hat{x}$ and the aggregation condition holds for $\riskregion$
  then $\hat{x}$ is also an optimal solution for the true problem.
\end{corollary}

\begin{proof}
  We prove this only for \ref{eq:p1}. The proofs for the
  other problems are very similar. First note that \eqref{eq:approx-risk-risk}
  implies that $\cvar\left(-\hat{x}^{T}\psi_{\riskregion}(Y))\right) = \cvar\left(-\hat{x}^{T}Y)\right)$ by Proposition~\ref{prop:low-var}.
  Note also that since $\E{-x^{T}\psi_{\riskregion}(Y)} = \E{-x^{T}Y}$ for
  all $x\in\rd$ that $\hat{x}$ is feasible with respect to the true problem.
  Now, if $\tilde{x}$  is an optimal solution to the true problem then,
  \begin{align*}
    \cvar(-\tilde{x}^{T}Y)) &\geq \cvar\left(-\tilde{x}^{T}\psi_{\riskregion}(Y)\right) \qquad \text{by Proposition~\ref{prop:low-var}}\\
   & \geq \cvar\left(-\hat{x}^{T}\psi_{\riskregion}(Y))\right) \qquad \text{by definition of } \hat{x}\\
   & = \cvar(-\hat{x}^{T}Y)).
  \end{align*}
  Hence, $\hat{x}$ is optimal with respect to the true problem.
\end{proof}

This result guarantees a certain robustness to misspecification
of the approximate risk region. Although checking whether an
approximate risk region is valid for a decision could in principle
be used as an optimality check, we will not use it in this way
as checking directly condition \eqref{eq:approx-risk-risk}
may be difficult. As will be seen in Section~\ref{sec:portfolio-case}
we instead will rely on out-of-sample testing to verify the quality
of a solution.

\subsection{Risk regions for elliptical distributions}
\label{sec:portfolio-recap-portfolio-risk}

In order to exploit risk regions for scenario generation one has to be
able to characterize these in a way which allows one to conveniently
test whether or not a point belongs to it.  In our previous paper, we
were able to do this in the case where the asset returns have
\emph{elliptical distributions}. Elliptical distributions are a
general class of distributions which include, among others,
multivariate Normal and multivariate $t$-distributions. See
\cite{FangKotzNg198911} for a full overview of the subject.

\begin{definition}[Elliptical Distribution]
  \label{def:sphere-ellipse}
  Let $X= (X_1, \ldots, X_d)$ be a random vector in $\rd$, then $X$ is said to be \emph{spherical}, if
  $$X \sim UX \qquad \text{for all orthonormal matrices } U$$
  where $\sim$ means the two operands have the same distribution function.

  Let $Y$ be a random vector in $\rd$, then $Y$ is said to be \emph{elliptical}
  if it can be written $Y = P^{T}X + \mu$ where $P\in \rr^{d\times d}$ is non-singular,
  $\mu\in\rd$, and $X$ is random vector with spherical distribution. Such
  an elliptical distribution will be denoted $\elliptical{X}{\mu}{P}$.
\end{definition}

This definition says that a random vector with a spherical
distribution is rotation invariant, and that an elliptical
distribution is an affine transformation of a spherical
distribution. Elliptical distributions are convenient in the context
of portfolio selection as we can write down exactly the distribution
of loss of a portfolio. In particular, if $Y\sim\elliptical{X}{\mu}{P}$
and $x\in\rd$ then
\begin{equation*}
  -x^T Y \sim \norm[Px]X_1 - x^T\mu,
\end{equation*}
where $\norm$ denotes the standard Euclidean norm and $X_{1}$ is the first component
of the spherical random vector $X$. Therefore, the $\beta$-quantile of the loss $-x^T Y$ is as follows:
$$F_{x}^{-1}(\beta) = \norm[Px] F^{-1}_{X_1}\left(\beta\right) - x^T\mu.$$
For $Y\sim\elliptical{X}{\mu}{P}$, we can thus rewrite the risk region in \eqref{eq:portfolio-riskregion-1} as follows:
\begin{equation}
  \label{eq:portfolio-riskregion-2}
  \mathcal{R}_{Y,\mathcal{X}}(\beta) := \bigcup_{x\in\mathcal{X}}\{y\in\rd : -x^T y \geq \norm[Px] F^{-1}_{X_1}\left(\beta\right) - x^T\mu\}.
\end{equation}
In this form it is still difficult to check whether a given point $\tilde{y}\in\rd$ belongs to it.
In \cite{Fairbrother15a} we provided a more convenient characterization of the risk region for
elliptical returns. This characterization makes use of the conic hull of the feasible region
$\mathcal{X}\subset\rd$.

\begin{definition}[Convex cones and conic hull]
  \label{def:convex-cone}
  A set $K \subset \rd$ is a cone if for all $x \in K$ and 
  $\lambda \geq 0$ we have $\lambda x \in K$. A cone is convex if for all 
  $x_1, x_2 \in K$ and $\lambda_1, \lambda_2 \geq 0$ we have 
  $\lambda_1 x_1 + \lambda_2 x_2 \in K$. The conic hull of a set 
  $\aggregion\subset\rd$ is the smallest convex cone containing $\aggregion$, and 
  is denoted $\conic{\aggregion}$.
\end{definition}

For example, suppose that our feasible
region consists of portfolios with non-negative investments (i.e. no
short-selling) and whose total investment is normalized to one,
that is:
\begin{equation*}
  \mathcal{X} = \{ x\in\rd: \sum_{i=1}^d x_i = 1,\ \ x_i \geq 0 \text{ for each } i = 1,\ldots, d\},
\end{equation*}
then the conic hull of this is the positive quadrant, that is $\conic{\mathcal{X}} = \rd_+$.
The alternative characterization also makes use of projections.

\begin{definition}[Projection]
  \label{def:proj}
  Let $C \subset \rd$ be a closed convex set, then for
  any point $y\in\rd$ we define its projection onto $C$
  to be the unique point $p_C(y)\in C$ such that
  \begin{equation*}
    \inf_{x\in C} \norm[x-y] = \norm[p_C(y) - y].
  \end{equation*}
\end{definition}

We are now ready to give a characterization of the risk region. For this we use the following convenient abuse of notation: for a set $\mathcal{A}\subset\rd$ and a matrix $T\in\rr^{d\times d}$, we write $T\left(\mathcal{A}\right) := \{ Ty : y\in\mathcal{A}\}$. The following result was proved in \cite{Fairbrother15a}.

\begin{theorem}
Suppose  $Y \sim \elliptical{X}{P}{\mu}$, $\mathcal{X}\subseteq\rd$ is convex and let $K = \conic{\mathcal{X}}$. Then the risk
region can be characterized exactly as follows:
\begin{equation}
  \label{eq:agg-region-cone-constraint}
  \riskregion_{Y,\mathcal{X}}(\beta) = P^{T}\left(\{\tilde{y}\in\rd: \norm[p_{K'}(\tilde{y} - \mu)] \geq F_{X_1}^{-1}\left(\beta\right)\}\right),
\end{equation}
where $K' = PK$ is a linear transformation of the conic hull $K$.
\end{theorem}

When $\mu = 0$ and $P = I_{d}$ the risk region can be simplified to $\riskregion_{Y,\mathcal{X}}(\beta) = \left(\{\tilde{y}\in\rd: \norm[p_{K}(y)] \geq F_{X_1}^{-1}\left(\beta\right)\}\right)$. This allows us to interpret the projection of $y$ onto $K$ as the portfolio which leads to the largest loss.

\subsection{Testing membership to a risk region}
\label{sec:portfolio-conic}

Scenario generation algorithms which exploit risk regions rely on the ability to test membership of the risk region for randomly sampled points.
The characterization of the risk region for portfolio selection problems
given in \eqref{eq:agg-region-cone-constraint} relies on one being able to calculate
the conic hull of the set of feasible portfolios, and also the ability
to project points onto a transformation of this.
In Section~\ref{sec:portfolio-conic-hull} we show how one can find the conic hull of the feasible region for typical constraints of a portfolio selection
problem. This conic hull is a \emph{finitely generated cone}. In Section~\ref{sec:portfolio-conic-project} we show how one can project points onto this type of cone.
Finally in Section~\ref{sec:computational-issues} we briefly discuss the computational issues for the membership tests.

\subsubsection{Conic hull of feasible region}
\label{sec:portfolio-conic-hull}

In portfolio problems, the feasible region is usually defined by linear constraints,
that is $\mathcal{X} = \{ x\in\rd: Ax \leq b \}$, where $A\in \rr^{m\times d}$ and $b\in\rm$. That is, the feasible region is the intersection
of a finite number of half-spaces. It is a well-known fact that any such intersection
can be written as the convex hull of a finite number of points plus the conical combination
of some more points (see Theorem 1.2 in \cite{Ziegler200804} for example). 
That is, there exists $x_1, \ldots, x_k\in\rd$ and $y_1, \ldots, y_l\in\rd$ such that
\begin{equation}
  \label{eq:polytope}
  \mathcal{X} = \{ \sum_{i=1}^k \lambda_i x_i + \sum_{j=1}^l \nu_j y_j:  \lambda, \nu \geq 0,\ \sum_{i=1}^k \lambda_i = 1\}.
\end{equation}
The conic hull of this region is the following \emph{finitely generated cone}:
\begin{equation*}
  \conic{\mathcal{X}} = \{ \sum_{i=1}^k \lambda_i x_i + \sum_{j=1}^l \nu_j y_j:  \lambda, \nu \geq 0\}.
\end{equation*}
To express the intersection of half-spaces in the form \eqref{eq:polytope},
we could use \emph{Chernikova's algorithm} (also known as the double description method)
\cite{Chernikova65,LeVerge92}.
Every finitely generated cone can also be written as a \emph{polyhedral cone},
that is, of the form $\{x\in\rd: Dx \geq 0\}$, and vice versa
(see \cite[Chapter 1]{Ziegler200804}). Chernikova's algorithm again provides a concrete
method for going between these two different representations. Although these two representations
are mathematically equivalent, as we shall see, they are algorithmically different.

We will suppose the constraints for our portfolio selection problem have the following form:
\begin{equation}
  \label{eq:portfolio-feasible}
  \mathcal{X} = \left\{ x \in \rd :
    \begin{split}
      \mathbf{1}^T x &= c\\
      a_i^Tx &\leq b_i \qquad \text{for } i = 1,\ldots,m,\\
      x &\geq 0,
    \end{split}
    \right\}
\end{equation}
where $\mathbf{1}$ is column vector of ones and $c > 0$. 
The first of these constraints
specifies the total of amount of capital to be invested, the
inequalities represent other constraints such as
quotas on the amount one can invest in a specific company or
industry. In this case, we can describe immediately the conic hull as a
polyhedral cone.

\begin{proposition}
  Let $\mathcal{X}$ be the set defined in \eqref{eq:portfolio-feasible} and let
  \begin{equation*}
    \mathcal{Y} = \left\{ x \in \rn :\left(\frac{b_i}{c}\mathbf{1} - a_i\right)^T x \geq 0 \text{ for } i = 1,\ldots, m,\ x \geq 0 \right\}
    \end{equation*}
\end{proposition}
then $\conic{\mathcal{X}} = \mathcal{Y}$.
\begin{proof}
  Given that $\mathcal{X}$ is convex, to show that $\conic{\mathcal{X}} = \mathcal{Y}$, it suffices
  to show that
  \begin{equation*}
    x \in \mathcal{Y}\setminus{\{0\}} \Longleftrightarrow\ \exists\ \lambda > 0 \text{ such that } \lambda x \in \mathcal{X}.
  \end{equation*}
  We demonstrate first the forward implication. Suppose $x\in\mathcal{Y}\setminus{\{0\}}$. Then,
  given that $x > 0$, we must have $v := \mathbf{1}^Tx > 0$. Then, setting $\lambda = \frac{c}{v}$, we have
  \begin{equation*}
    \mathbf{1}^T (\lambda x) = v \frac{c}{v} = c.
  \end{equation*}
  Since $\mathcal{Y}$ is a cone, we have $\lambda x \in \mathcal{Y}$, hence
  \begin{align*}
    & (\frac{b_i}{c}\mathbf{1} - a_i)^T \frac{c}{v} x\ \geq\ 0 \\
    \therefore\qquad & \frac{c}{v} a_i^T x\ \leq\ \frac{b_i}{c} \frac{c}{v} \underbrace{\mathbf{1}^T x}_{=v}\\
    \therefore\qquad &  a_i^T (\frac{c}{v}x)\ \leq\ b_i
  \end{align*}
  and so $\lambda x \in \conic{\mathcal{X}}$.
  
  We now prove the backwards implication. Suppose $x\in\conic{\mathcal{X}}\setminus \{ 0 \}$. Then there exists
  $\lambda > 0$ such that $\lambda x \in \mathcal{X}$, that is
  \begin{align*}
    \mathbf{1}^T \lambda x &= c\\
    a_i^T \lambda x &\leq b_i
  \end{align*}
  Therefore,
  \begin{align*}
    &\frac{a_i^T \lambda x}{\mathbf{1}^T\lambda x}\ \leq\ \frac{b_i}{c}\\
    \text{and so}\qquad & \left(\frac{b_i}{c}\mathbf{1} - a_i\right)^T x\ \geq\ 0.
  \end{align*}
  Hence $x\in \mathcal{Y}$ as required.
\end{proof}

\begin{figure}[h]
  \centering
  \includegraphics[width=0.4\textwidth]{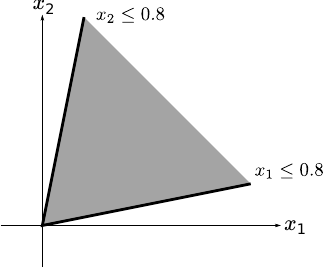}
  \caption{Conic hull from simple quota constraints given $x_1 + x_2 = 1$ and $x_1, x_2 \geq 0$}
  \label{fig:cone_constraints}
\end{figure}

Figure \ref{fig:cone_constraints} shows how simple constraints in $\rr^2$ affect
the conic hull of the feasible region given the total
investment and positivity constraints. 

\subsubsection{Projection onto a finitely generated cone}
\label{sec:portfolio-conic-project}

First, suppose that we can represent the conic hull of the feasible
region $\mathcal{X}\subset\rd$ as a finitely generated cone with $k$ generators, that is $K = \{ Ay
: y \geq 0\}$ where $A \in \rr^{k\times d}$.  By definition, the
projection of a point $x_0\in\rd$ can be found by solving the
following quadratic program:
\begin{equation}
  \label{eq:finite-project}
  \minimize[y\geq 0]\ \norm[Ay - x_0]_2^2
\end{equation}
In particular, if $y^*$ is the optimal solution then $p_K(x_0) =
Ay^*$.  By formulating the KKT conditions
\cite[Chapter~5]{boyd2004convex} of this problem, it can be seen that
this problem is equivalent to solving the following linear
complementarity problem (LCP):
\begin{align*}
  \text{Find } y, z\in \rd \text{ such that }\\
  z - A^TAy &= -A^Tx_{0}\\
  z^Ty &= 0\\
  y,z &\geqslant 0.
\end{align*}
If $(y,z)$ is a solution to the above problem, then the required
projection is $p_{K}(x_{0}) = Ay$. LCPs can be solved by more
specialized algorithms than standard quadratic programs
such as Lemke's algorithm \cite{cottle1992linear}.

Now, suppose instead we have a polyhedral characterization of the conic hull, that is a cone of the form:
\begin{equation}
  \label{eq:poly-cone}
  K = \{x\in\rd : Bx \geq 0\}.
\end{equation}
The projection of a point $x_0\in\rd$ onto the polyhedral cone in
\eqref{eq:poly-cone} is the solution of the following quadratic
program:
\begin{align*}
  \minimize[x] &\norm[x - x_0]_2^2\\
  \text{subject to } & Bx \geq 0.
\end{align*}

Although the former problem in \eqref{eq:finite-project} can often be
solved more efficiently using specialized algorithms, we will in practice use
both approaches. For conic hulls with a small
number of extremal rays, for example $K=\rr_+^d$ we will use use the former method.
As we add more constraints to the problem, we have found from experience
that the number of extremal rays can exponentially increase, which for the
former approach leads to cumbersomely large LCP problems. In this case
we will use the polyhedral representation for projection.

\subsubsection{Computational issues}
\label{sec:computational-issues}

Given that testing membership of the risk region (for elliptically distributed
returns) involves solving a small LCP or quadratic program, using 
this methodology could potentially become computationally expensive,
especially if used to construct large scenario sets for high dimensional
problems. However, this issue can be
mitigated in a few ways. Firstly, the membership test for a point can
be conducted independently from that of another point, which means
that membership tests for a large number of points is naturally
parallelizable. Secondly, for the case where the case $K \subseteq \rd_{+}$
the loss function $y \mapsto -x^{T}y$ is monotonic. Therefore,
if we have set of points $y_{1}, \ldots, y_{k}$ which are in
the risk region, we know that a point $\tilde{y}$ is also
in the risk region if it is dominated by any of these points.
A similar shortcut exists for testing if $\tilde{y}$ is in the non-risk region.
Finally, the membership test could also be made more efficient,
by directly testing the condition $\norm[p_K (y)] \leq \alpha $ without calculating the full
projection $p_K (y)$. For example, the quadratic program used to calculate the projection could be solved only to an accuracy sufficient to test this condition. This
could be easily implemented through a callback function in the quadratic
program solver.

\section{Scenario generation}
\label{sec:portfolio-scengen}


In this Section we show how risk regions can be exploited for the
purposes of scenario generation. In Section~\ref{sec:portfolio-scengen-agg-sampling-reduction} we present two specific methods which work essentially by prioritizing
the construction of scenarios in the risk region.
In Section~\ref{sec:portfolio-scengen-ghost} we 
propose a new heuristic algorithm based on the SAA method \cite{KleywegtEA01}.
This heuristic boosts the performance of the proposed sampling algorithm
through the addition of artificial constraints to the problem.

\subsection{Aggregation sampling and reduction}
\label{sec:portfolio-scengen-agg-sampling-reduction}

In \cite{Fairbrother15a} we proposed two methods to exploit risk
regions. The first of these allows the user to specify the final
number of scenarios in advance.  The algorithm, which is called
\emph{aggregation sampling}, samples scenarios, aggregating all
samples in the non-risk region and keeping all in the risk region,
until we have the required number of risk scenarios, that is the
required number of scenarios in the risk region. This is described in
Algorithm~\ref{alg:agg-sampling}.

\begin{algorithm}

  \textnormal{}\SetKwData{Left}{left}\SetKwData{This}{this}\SetKwData{Up}{up}
  \SetKwFunction{Union}{Union}\SetKwFunction{FindCompress}{FindCompress}
  \SetKwInOut{Input}{input}\SetKwInOut{Output}{output}
  
  \Input{$\riskregion \subset \rd$ approximate risk region, $N_{\riskregion}$ number of required risk scenarios}
  \Output{$\{(y_s, p_s)\}_{s=1}^{N_{\riskregion} + 1}$ scenario set}
  $n_{\riskregion^c} \leftarrow 0$,\  $n_{\riskregion} \leftarrow 0$,\ $y_{\riskregion^{c}}=\mathbf{0}$\;
  \While{$n_{\riskregion} < N_{\riskregion}$} {
    Sample new point $y$\;
    \If{$y\in\riskregion$}{
      $n_{\riskregion} \leftarrow n_{\riskregion} + 1 $; $y_{n_{\riskregion}} \leftarrow y$\;
    }
    \Else{
      $n_{\riskregion^c} \leftarrow n_{\riskregion^c} + 1$;\ 
      $y_{\riskregion^{c}} \leftarrow \frac{1}{n_{\riskregion^{c}}+1}\left(n_{\riskregion^{c}}y_{\riskregion^{c}} + y\right) $
    }
  }
  \lForEach{$i$ in $1, \ldots, N_{\riskregion}$}{$p_i \leftarrow \frac{1}{\left(n_{\riskregion^c} + N_{\riskregion}\right)}$}
  \If{$n_{\riskregion^c} > 0$}{
    $p_{n_{\riskregion^c} + 1} \leftarrow \frac{n_{\riskregion^{c}}}{n_{\riskregion^c}+ \mathcal{N}_{\riskregion}}$\;
  }
  \Else{
    Sample new point $y$\;
    $n_{\riskregion^c} \leftarrow 1$;
    $y_{N_{\riskregion} + 1} \leftarrow y$\;
  }
  
  $p_{N_{\riskregion}+1} \leftarrow \frac{n_{\riskregion^c}}{n_{\riskregion^c} + N_{\riskregion}}$
   \caption{Aggregation sampling}
  \label{alg:agg-sampling}
\end{algorithm}

Let $q=\Prob{Y\in\riskregion_{Y,\mathcal{X}}^{c}}$ be the probability
of the non-risk region, and $n$ be the number of risk scenarios
required.  Define $N(n)$ to be the effective sample size from
aggregation sampling, that is, the number of draws until the algorithm
terminates\footnotemark.
\footnotetext{For simplicity of exposition we discount the event that the while-loop of the algorithm terminates with $n_{\riskregion^{c}} = 0$ which occurs with probability $q^{n}$}
The quantity $N(n)$ is a random variable:
\begin{equation*}
  N(n) \sim n + \mathcal{NB}(n, q), 
\end{equation*}
where $\mathcal{NB}(N,q)$ denotes a \emph{negative binomial} random variable. Recall
that a negative binomial random variable $\mathcal{NB}(n, q)$ is the number
of failures in a sequence of Bernoulli trials with probability of success $q$
until $n$ successes have occurred. The expected effective sample size of aggregation
sampling is thus as follows:
\begin{equation*}
  \E{N(n)} = n + n \frac{q}{1-q}
\end{equation*}
The expected effective sample size can be thought of as the sample
size required for basic sampling to produce the same number of
scenarios in the risk region. The difference between the desired
number of risk scenarios, and expected effective sample size is
proportional to the ratio $\frac{q}{1-q}$. In particular, as the
probability of the non-risk region approaches one, this gain tends to
infinity.

The converse to aggregation sampling is sampling a set of a given size
$n$ and then aggregating all scenarios in the risk region of the
underlying distribution. We call this \emph{aggregation reduction}. 
This can be viewed as a sequence of $n$
Bernoulli trials, where success and failure are defined in the same
way as described above. The number of scenarios in the reduced sample,
$R(n)$ is as follows:
\begin{equation*}
  R(n) \sim n - \mathcal{B}(n, q) + 1
\end{equation*}
where $\mathcal{B}(n,q)$ denotes a binomial random variable. The expected reduction in
scenarios in aggregation reduction is thus $nq -1$.

The reason why aggregation sampling and aggregation reduction work is
that, for large samples, they are equivalent to sampling from the
aggregated random vector. Suppose that $Y_{1}, Y_{2}, \ldots$ is a
sequence of independently identically distributed (i.i.d.) random
vectors with the same distribution as $Y$, then
$\psi_{\riskregion}(Y_1),\psi_{\riskregion}(Y_2),\ldots$ is a sequence
of i.i.d. random vectors with the same distribution as the aggregated
random vector $\psi_{\riskregion}(Y)$. Denote by
$\tilde{\rho}_{n,\beta}(x)$ the value of the tail-risk measure for the
decision $x\in\mathcal{X}$ for the sample
$\psi_{\riskregion}(Y_1),\ldots,\psi_{\riskregion}(Y_n)$, and by
$\hat{\rho}_{n,\beta}$ the analogous function the scenario set
constructed by aggregation sampling. The following result, adapted
from \cite{Fairbrother15a} for the portfolio selection problem, gives
precise conditions under which aggregation sampling is asymptotically valid.

\begin{theorem}
  \label{thm:agg-sample-consistent}
  Suppose the following conditions hold:
  \begin{enumerate}[(i)]
  \item For each $x\in\mathcal{X}$, $F_{x}$ is strictly increasing and continuous in some neighborhood of $F_{x}^{-1}(\beta)$
  \item $\E{\ Y | Y\in\riskregion^c\ } \in \interior{\riskregion_{\mathcal{X}}^c}$
  \item $\mathcal{X}$ is compact.
  \end{enumerate}
  Then, with probability 1, for $n$ large enough $\tilde{\rho}_{n, \beta} \equiv \hat{\rho}_{N(n), \beta}$.
\end{theorem}

See \cite[Section~4]{Fairbrother15a} for a full
proof of the consistency of these algorithms.

\subsection{Ghost constraints}

\label{sec:portfolio-scengen-ghost}

We noted above that the performance of our methodology improves
as the probability of the non-risk region increases.
In particular, the expected effective sample size in aggregation
sampling increases as the probability of the non-risk risk region increases.
Now, by its definition \eqref{eq:portfolio-non-risk-region} the
non-risk region grows as the problem becomes more constrained.
This suggests that it may be helpful to add constraints to our
problem which shrink the set of feasible portfolios, but which are not
themselves active, in the sense that their presence does not affect
the set of optimal solutions. We will refer to a constraint added to a
problem to boost the performance of our methodology, loosely, as a
\emph{ghost constraint}.

Finding non-active constraints to add to our problem is non-trivial as 
it relies on some knowledge of the optimal solution set. Moreover, even
verifying whether or not a particular constraint is active is difficult in
general for stochastic programs. For a deterministic objective
function which is convex and for which all constraints are convex (and
the optimal solution is unique) a constraint $\{x: g(x)\leq 0 \}$ is
active if and only if it is binding at the optimal solution
$x^{*}$, that is $g(x^{*}) = 0$.  For a stochastic program, we are
typically solving a scenario-based approximation and so a constraint
which is not binding with respect to the scenario-based approximation
may be binding with respect to the true problem and vice versa.
A rigorous test of whether a ghost constraint is active in the sense
above is beyond the scope of this paper.  We simply promote the idea
here that ghost constraints may be a useful way of finding better
solutions.

We resort to heuristic rules to choose ghost constraints. For
example, one could constrain our set of feasible portfolios to some
neighborhood of a good quality solution. This suggests an iterative
procedure whereby one samples scenario sets using aggregation
sampling, solves the resulting problems, adjusts the problem constraints
and then resamples. In Algorithm~\ref{alg:saa-ghost-sampling} we propose
such a heuristic procedure based on the sample average approximation (SAA) method
of \cite{KleywegtEA01}. We call this procedure the \emph{SAA method
with ghost constraints}.  Like with the original SAA method, this
algorithm is presented in a very general form, as the update rules,
such as how the bounds are adjusted in each iteration, can be
implemented in many different ways. In
Section~\ref{sec:portfolio-case} we test this algorithm on a realistic
and difficult problem.

\begin{algorithm}
  \caption{Sample average approximation method with ghost constraints}
  \SetKwData{Left}{left}\SetKwData{This}{this}\SetKwData{Up}{up}
  \SetKwFunction{Union}{Union}\SetKwFunction{FindCompress}{FindCompress}
  \SetKwRepeat{Do}{do}{until}
  \SetKwInOut{Input}{input}\SetKwInOut{Output}{output}
  
  Initialise $l_{i} = -\infty, u_{i} = \infty$ for each $i = 1,\ldots,d$\;
  \Do{Optimality gap and variance are sufficiently small for some $m$}{
    Add constraints $l \leq x \leq u$ to problem\;
    Construct risk region $\riskregion$ for problem\;
    \For{$i = 1,\ldots,M$}{
      Generate a sample of size $N$ using aggregation sampling with risk region $\riskregion$\;
      Solve corresponding problem with objective value $\nu_{N}^{m}$ and optimal solution $\hat{x}_{N}^{m}$\;
      Estimate optimality gap of solution and the variance of the optimality gap\;
    }
    Increase size of $N$\;
    Adjust bounds $l$ and $u$\;
  }
  Choose the best solution $\hat{x}$ among all candidate solutions $x_{N}^{m}$ produced, using a screening and selection procedure.
  \label{alg:saa-ghost-sampling}
\end{algorithm}

\section{Probability of the non-risk region}
\label{sec:portfolio-probnonrisk}

The benefit of aggregation sampling and reduction depends on the
probability of the non-risk region. As was observed in
\cite{Fairbrother15a} the probability of the non-risk region tends to
decrease as the problem dimension increases, but increases as we
tighten our problem constraints, and as we increase $\beta$, the level
of the tail risk measure. In this section we make some empirical
observations on how this probability varies with heaviness of the
tails, and the correlations of the distribution.

The first observation is that in the presence of positivity constraints,
the probability of the non-risk region increases as the the
correlation between random variables increases. This can be seen in Figure \ref{fig:portfolio-corprob} which plots the probability of the risk region as a function of correlation for some two-dimensional distributions.
An intuitive explanation for this type of behavior is that
in the case of positive correlations there is much more overlap in the risk regions of the individual
portfolios. 

\begin{figure}[h]
  \centering
  \includegraphics[width=0.6\textwidth]{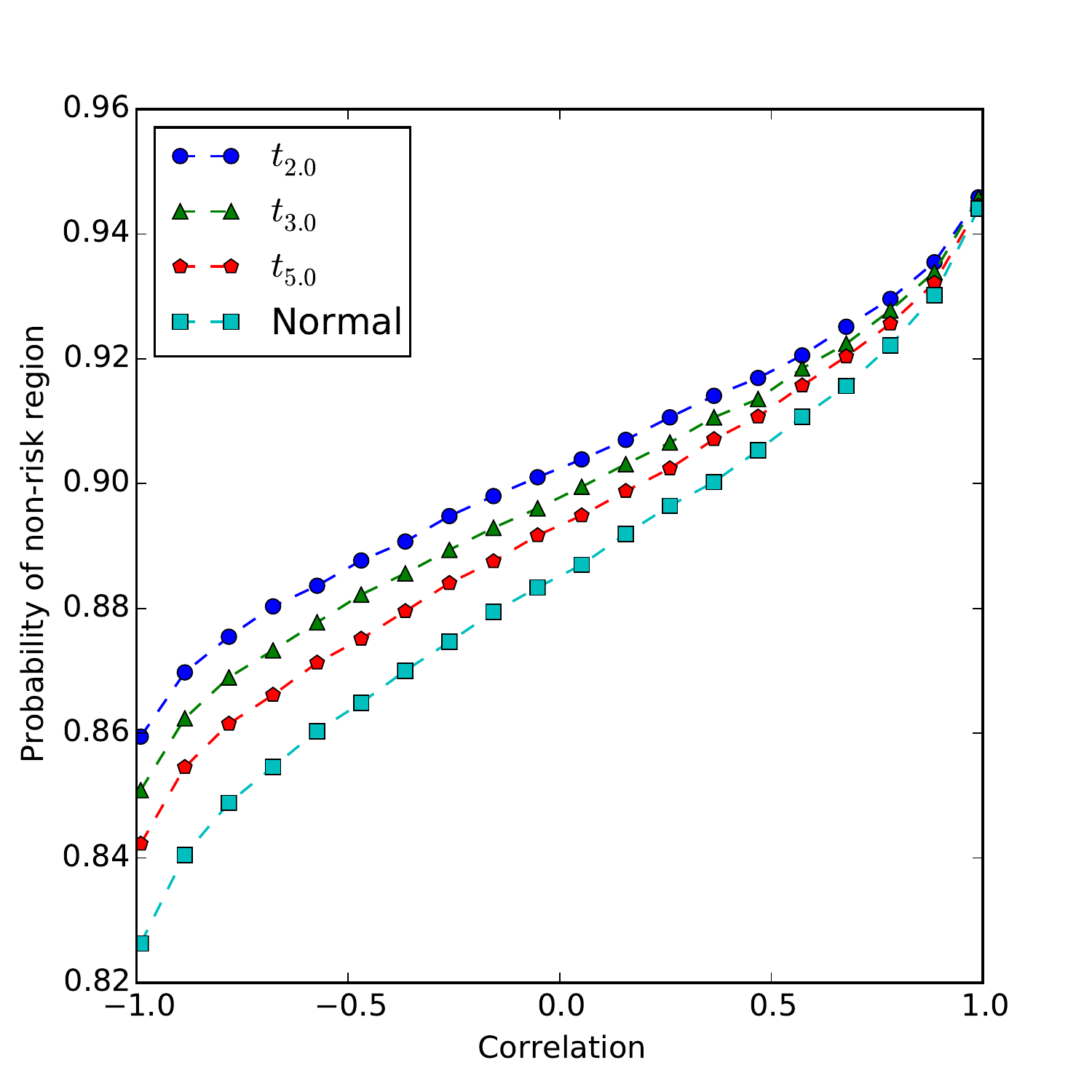}
  \caption{Correlation vs. Probability of non-risk region for some 2-dimensional elliptical distributions, positivity constraints and $\beta = 0.95$}
  \label{fig:portfolio-corprob}
\end{figure}

The extent to which probabilities vary with correlation seems
to be much greater in higher dimensions. In Figure \ref{fig:portfolio-cor-dim-plot}
we have plotted for Normal returns and a range of dimensions, the probabilities
of the non-risk region for a particular type of correlation matrix: $\Lambda\left(\rho\right) \in \rr^{d\times d}$
where $\Lambda(\rho)_{ij} = \rho$ for $i\neq j$ and $\rho > 0$. In the case of $\rho = 0$, the probability decays very
quickly to zero as the dimension increases, whereas as when $\rho$ is close to one, the probability of the non-risk
region approaches $\beta$ for all dimensions.

\begin{figure}[h]
  \centering
  \includegraphics[width=0.6\textwidth]{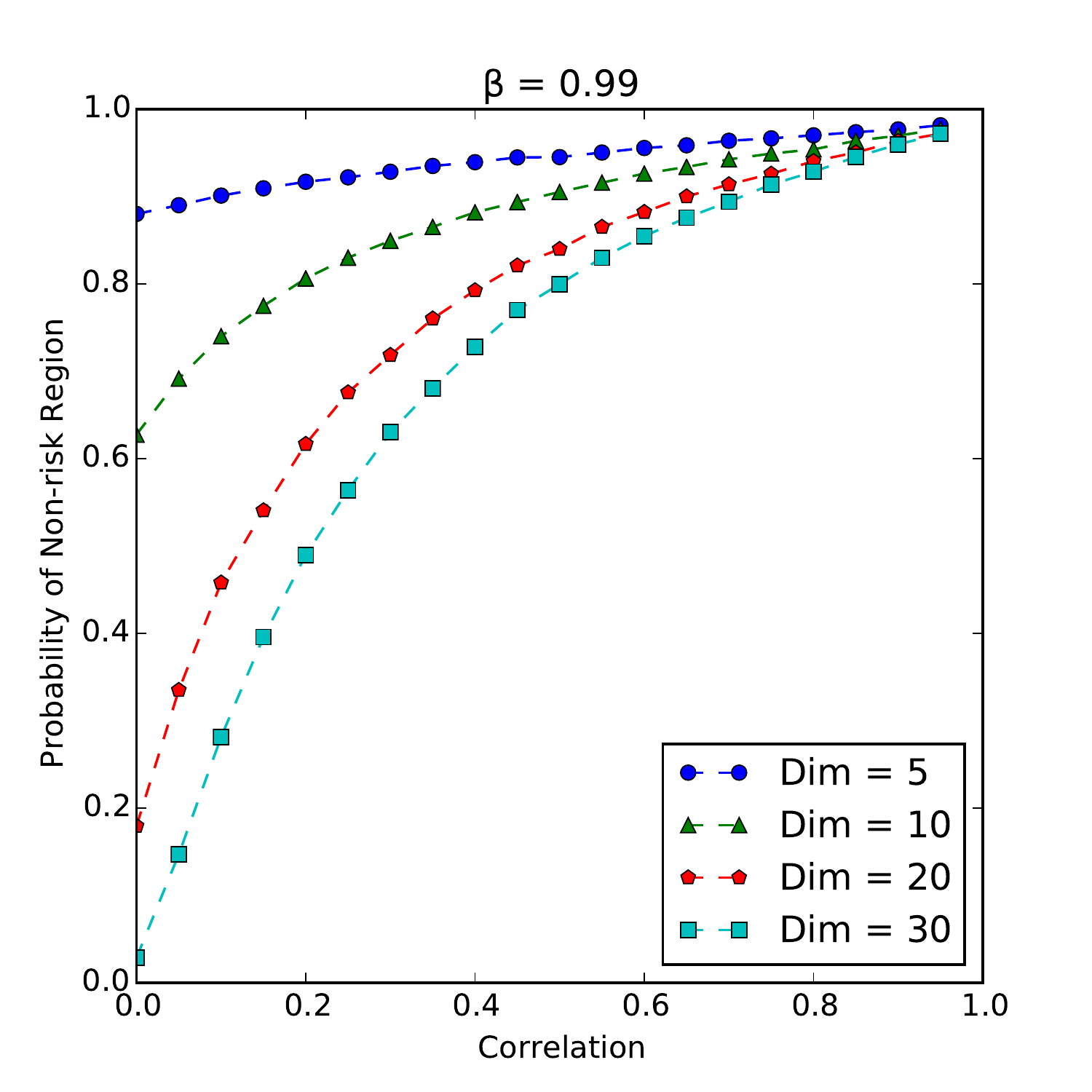}
  \caption{Probability of non-risk region for a range of correlation matrices and dimensions for Normal returns}
  \label{fig:portfolio-cor-dim-plot}
\end{figure}

Our next observation is that the probability of the non-risk region
seems to increase as the tails of the distributions become heavier. 
In Figure \ref{fig:portfolio-heaviness} are plotted
the probabilities of risk regions for some spherical distributions and a range of
dimensions. Note that multivariate t-distributions have heavier tails
than the multivariate Normal distribution, but the tails get lighter
as the degrees of freedom parameter increases. This phenomenon can also be observed in Figure \ref{fig:portfolio-corprob}.

\begin{figure}[h]
  \centering
  \begin{subfigure}[b]{0.45\textwidth}
    \includegraphics[width=\textwidth]{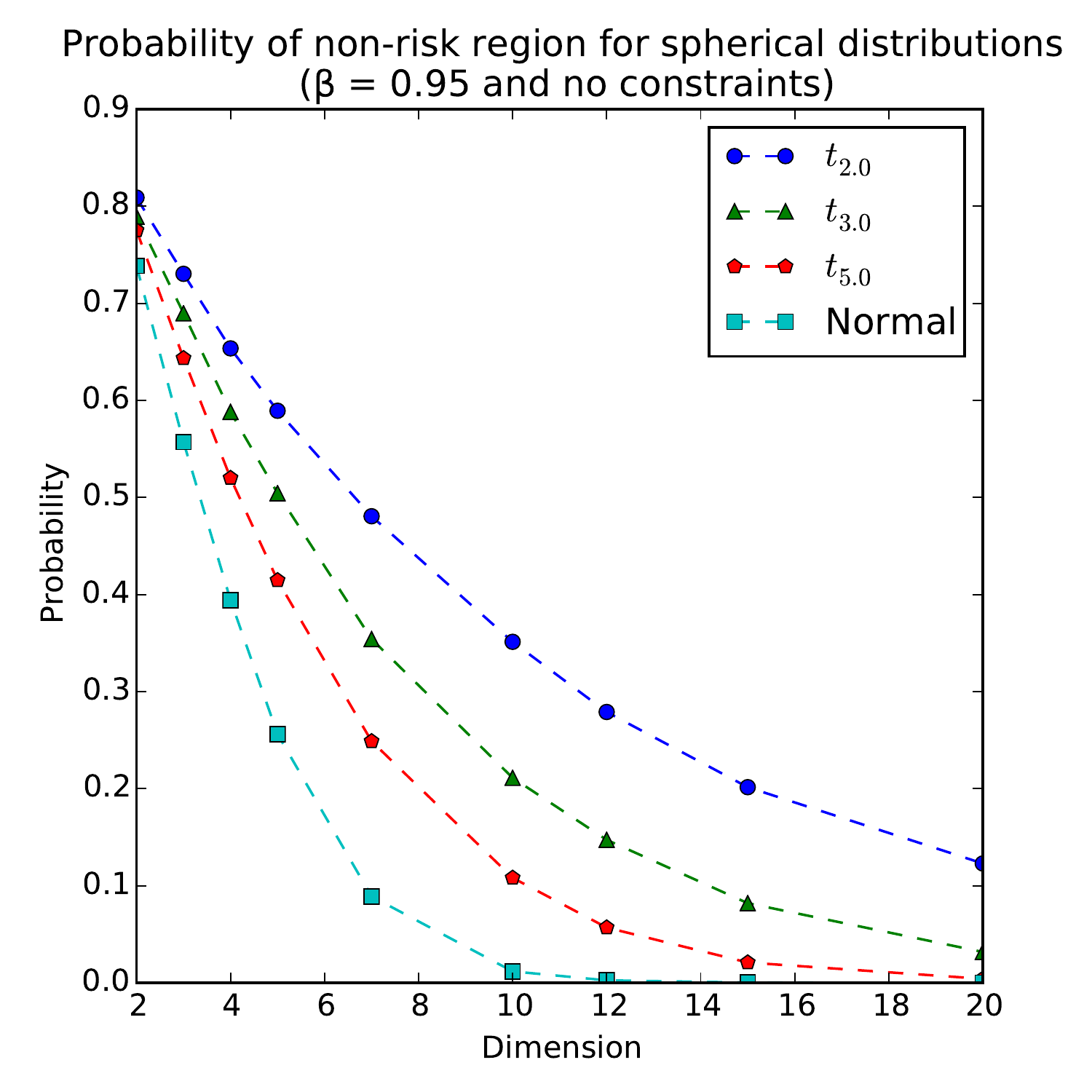}
  \end{subfigure}
  \begin{subfigure}[b]{0.45\textwidth}
    \includegraphics[width=\textwidth]{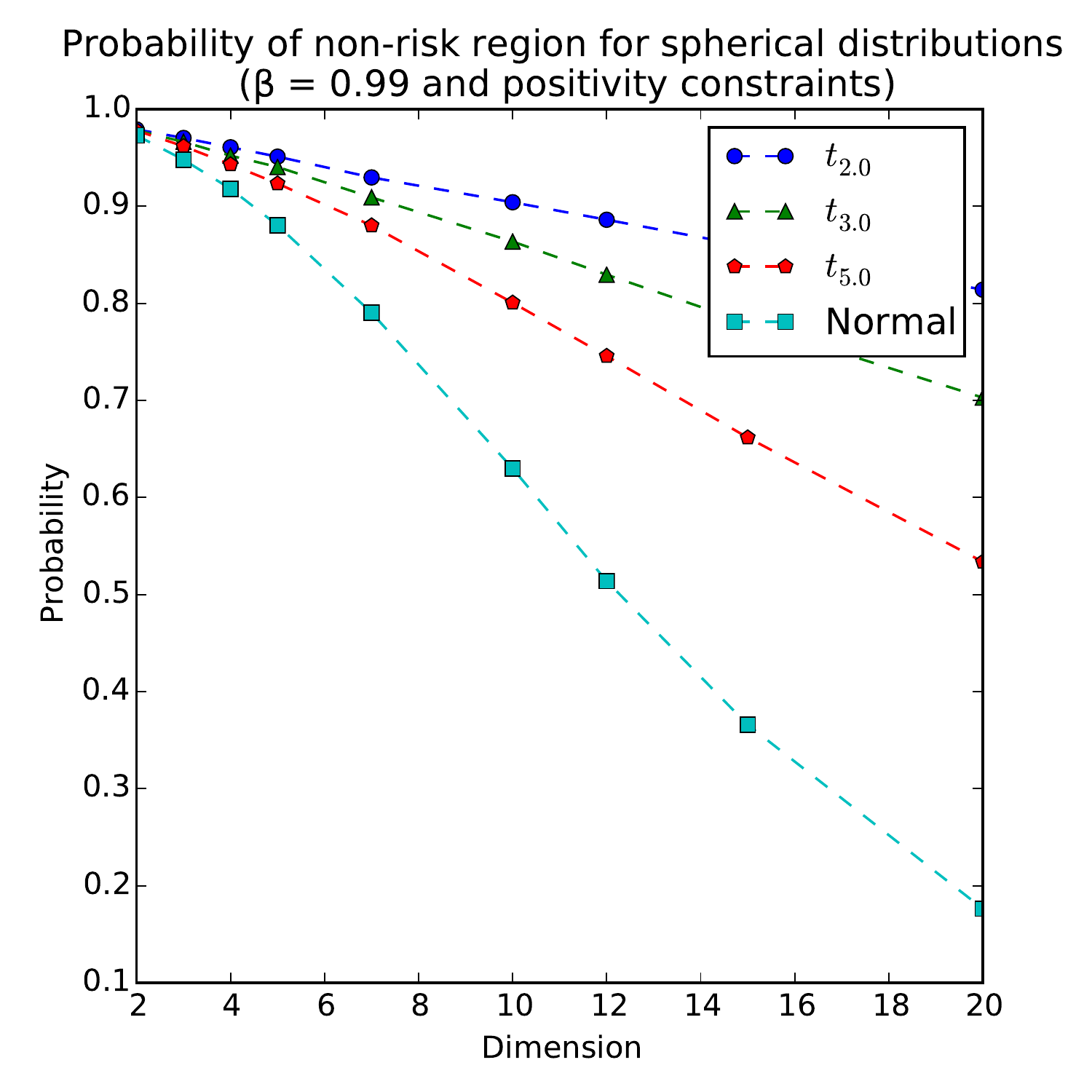}
  \end{subfigure}
  \caption{Probability of non-risk regions for different spherical distributions and dimensions}
  \label{fig:portfolio-heaviness}
\end{figure}

The observations made in this section suggest that that the
application of our methodology will be particularly effective when
applied to real stock data tend to be positively
correlated and have heavy tails.

\section{Numerical tests}
\label{sec:portfolio-numtests}

In this Section we test the numerical performance of our methodology
for realistic distributions. There are three parts to these tests:
the calculation of the probability of the non-risk region for a range
of distributions and constraints, the performance of aggregation
sampling, and the performance of aggregation reduction. To allow
us to measure the quality of the solutions yielded by the proposed methods,
the majority of the tests are performed for elliptical distributions on problems 
without integer variables. These problems can be solved exactly using other methods.
In Section~\ref{sec:portfolio-experimental-set-up} we describe
our experimental set-up, in particular we justify the distributions
constructed for these experiments. The remaining three sections
detail the individual experiments and their results.

\subsection{Experimental set-up}
\label{sec:portfolio-experimental-set-up}

For robustness we will use
several randomly constructed distributions for each family
of distribution and each dimension we are testing. We construct these by fitting parametric distributions
to real data. We use real data rather than arbitrarily generated problem
parameters for two reasons. Firstly, generating parameters which
correspond to well-defined distributions can be problematic.
For example, for the moment matching algorithm, there may
not exist a distribution which has a given set of target moments (see \cite{KlaassenEA00} and
\cite{JohnsonRogers51} for instance). Secondly, as was observed in 
Section~\ref{sec:portfolio-probnonrisk}, the probability of the non-risk region
can vary widely, and so it is most meaningful to test the performance of 
our methodology for distributions which 
are realistic for portfolio selection problems.

\begin{figure}[h]
  \centering
  \includegraphics[width=\textwidth]{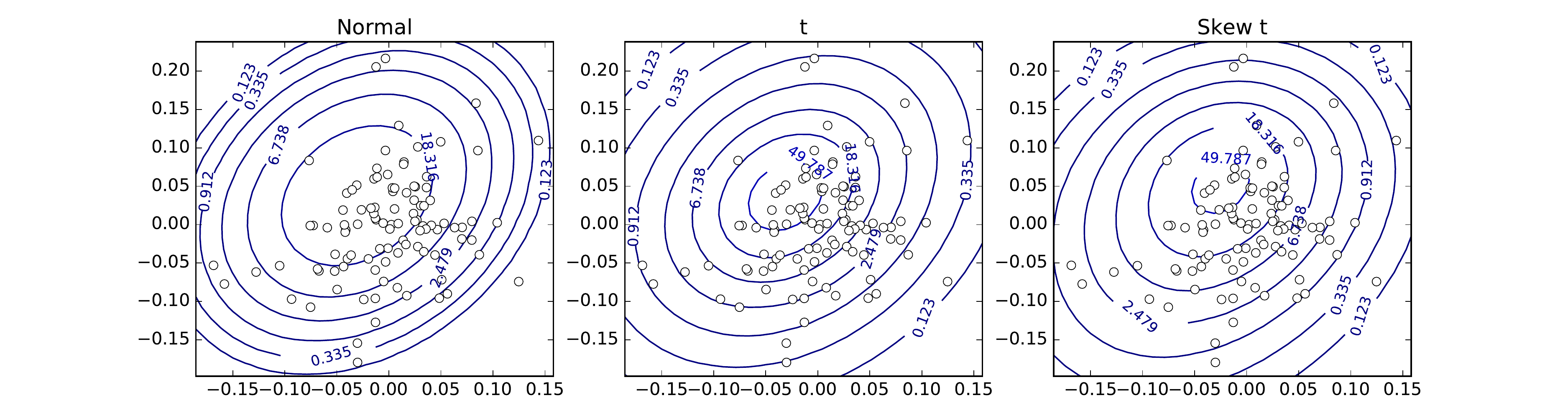}
  \caption{Contour plots of distributions fitted to financial return data for two assets}
  \label{fig:fitted-data}
\end{figure}

We construct our distributions from
monthly return data from between January 2007 and February 2015 for 90
companies in the FTSE 100 index. For each dimension in the test, we
randomly sample five sets of companies of that size, and for each of
these sets fit Normal, t distributions and Skew-t distributions to the
associated return data. Figure \ref{fig:fitted-data} shows for two
stocks the contours of the fitted density functions overlaying the
historical return data. For the $t$ distributions we fix the degrees
of freedom parameter to $4.0$. This is so that we can more easily
compare the effect of heavier tails on the results of our tests.
We allow the corresponding parameter for the skew-t distributions to
be fitted from the data.

These three distributions are fitted to the data
through maximum likelihood estimation, weighing
each observation equally; our aim here is not to construct distributions
which accurately capture the uncertainty of future returns, but to simply construct
distributions which are realistic for this type of problem.
We also use scenario sets constructed using the moment-matching
algorithm. For each random set of companies, we calculate all the required
marginal moments and correlations from their historical returns, and use
these as input to the moment-matching algorithm.
To allow us to compare results, the same constructed distributions
are used across the three sets of numerical tests.

Throughout this section we use the $\cvar$ as our tail risk measure.
This is not only because the $\cvar$ leads to tractable scenario-based
optimization problems, but also for elliptically distributed returns
we can evaluate the $\cvar$ exactly which provides us with a means to evaluate
the true performance of the solutions yielded by the approximate scenario-based
problems. In addition, to ensure that the non-risk region does not
have negligible probability, we will assume that we always have positivity
constraints on our investments (i.e. no short selling). 

The first class of non-elliptical distributions we use in this paper are
known as multivariate Skew-t distributions
\cite{AzzaCapi03}. This class of distributions generalizes the
elliptical multivariate t-distributions through the inclusion of an
extra set of parameters which regulate the skewness. In this case we
approximate the risk region with the risk region of the corresponding
t-distribution.

The second class we use are discrete distributions
constructed using the moment matching algorithm of \cite{HoylandEA00}. These
distributions have been applied previously to financial problems \cite{KautEA03}. This algorithm constructs
scenario sets with a specified correlation matrix and whose marginals have specified
first four moments. This algorithm works by first taking a sample from a multivariate Normal distribution, and then iteratively applying transformations to this
until the difference between its marginal moments and correlation matrix are
sufficiently close to their target values. Since the algorithm is initialized with a
sample from an elliptical distribution, the final distribution is near elliptical and we approximate the risk
region for these distributions with the risk region of a multivariate normal distribution with the same mean and covariance structure.

\subsection{Probability of non-risk region with quota constraints}
\label{sec:portfolio-numtests-prob}

We first estimate the probability of the non-risk region
for a range of distributions, dimensions and constraints. 
We calculate these probabilities only for the Normal and t distributions as skew-t distributions
and moment matching scenario sets use surrogate risk regions based on these distributions.
The main purpose of this is to provide intuition about under what circumstances the methodology
is effective: there is little to be gained from aggregating scenarios 
in a non-risk region of negligible probability.

For each distribution we sample 2000 scenarios and calculate the
proportion of points in the non-risk region for different levels of
$\beta$ and constraints. In particular, for each dimension we calculate for
$\beta = 0.95$, and $\beta=0.99$, and for a range of \emph{quotas}.
The feasible region corresponding to quota $0<q<1$ is $\{x\in\rd: 0 \leq x_i \leq q \text{ for } i = 1,\ldots,d,\ \sum_{i=1}^dx_i = 1\}$.
Quotas are quite a natural constraint to use in the portfolio selection
problem as they ensure that a portfolio is not overexposed to one asset.
The quotas may also be viewed as ghost constraints to be used in cases
where the probability of the non-risk region with only positivity
constraints is too small.

\begin{figure}[h]
  \centering
  \begin{subfigure}[b]{0.45\textwidth}
    \includegraphics[width=\textwidth]{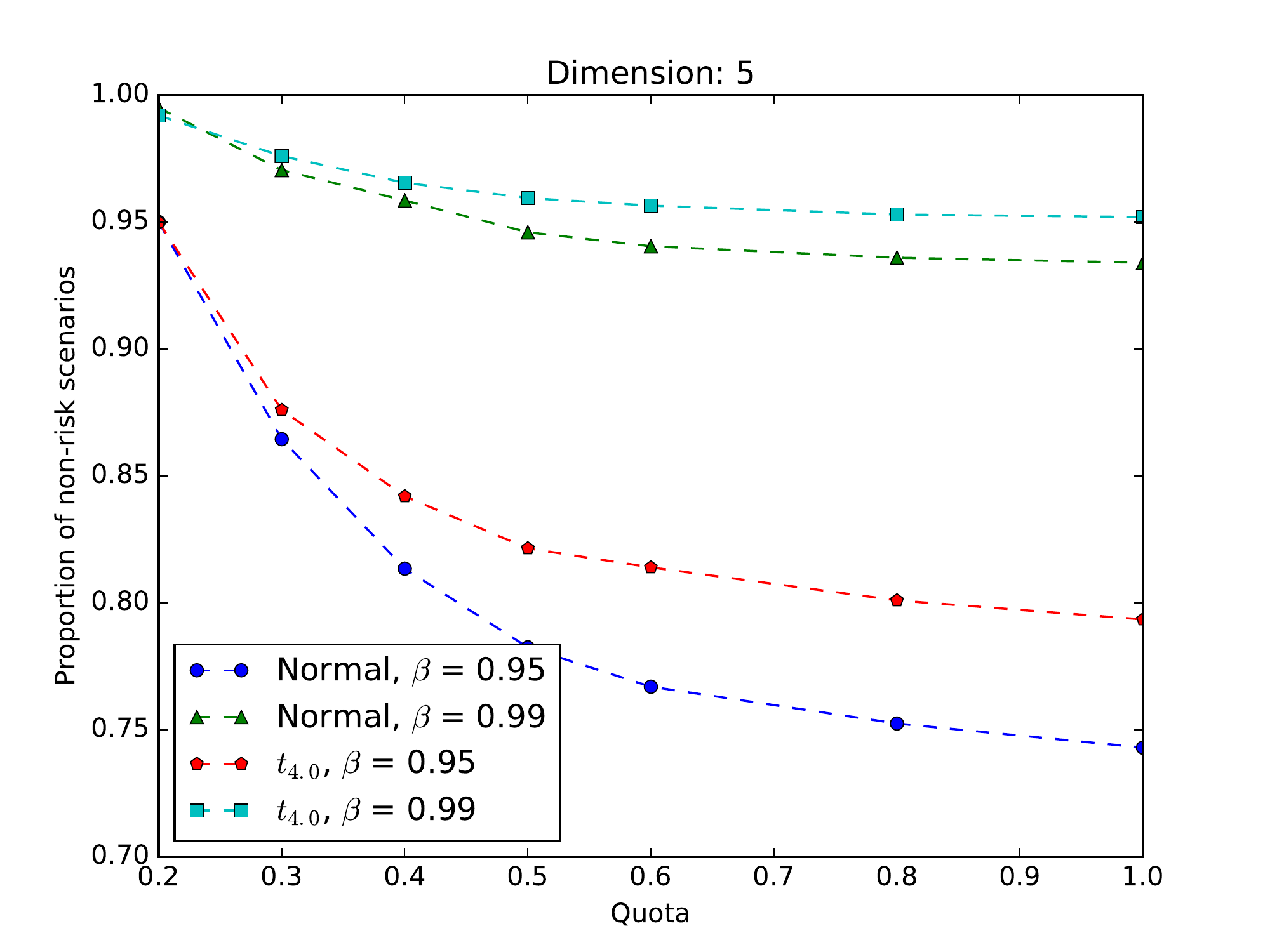}
  \end{subfigure}
  \begin{subfigure}[b]{0.45\textwidth}
    \includegraphics[width=\textwidth]{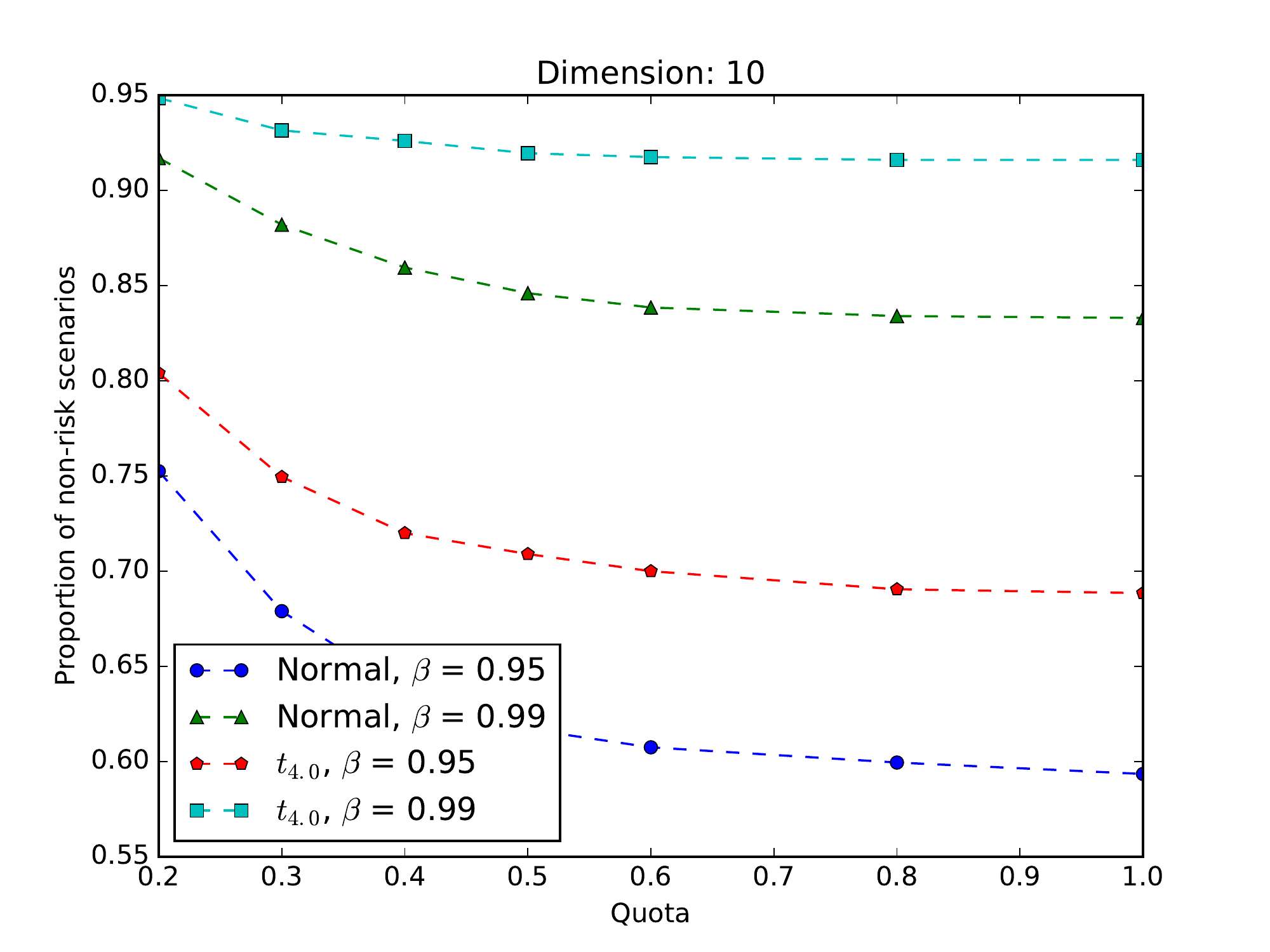}
  \end{subfigure}
  \begin{subfigure}[b]{0.45\textwidth}
    \includegraphics[width=\textwidth]{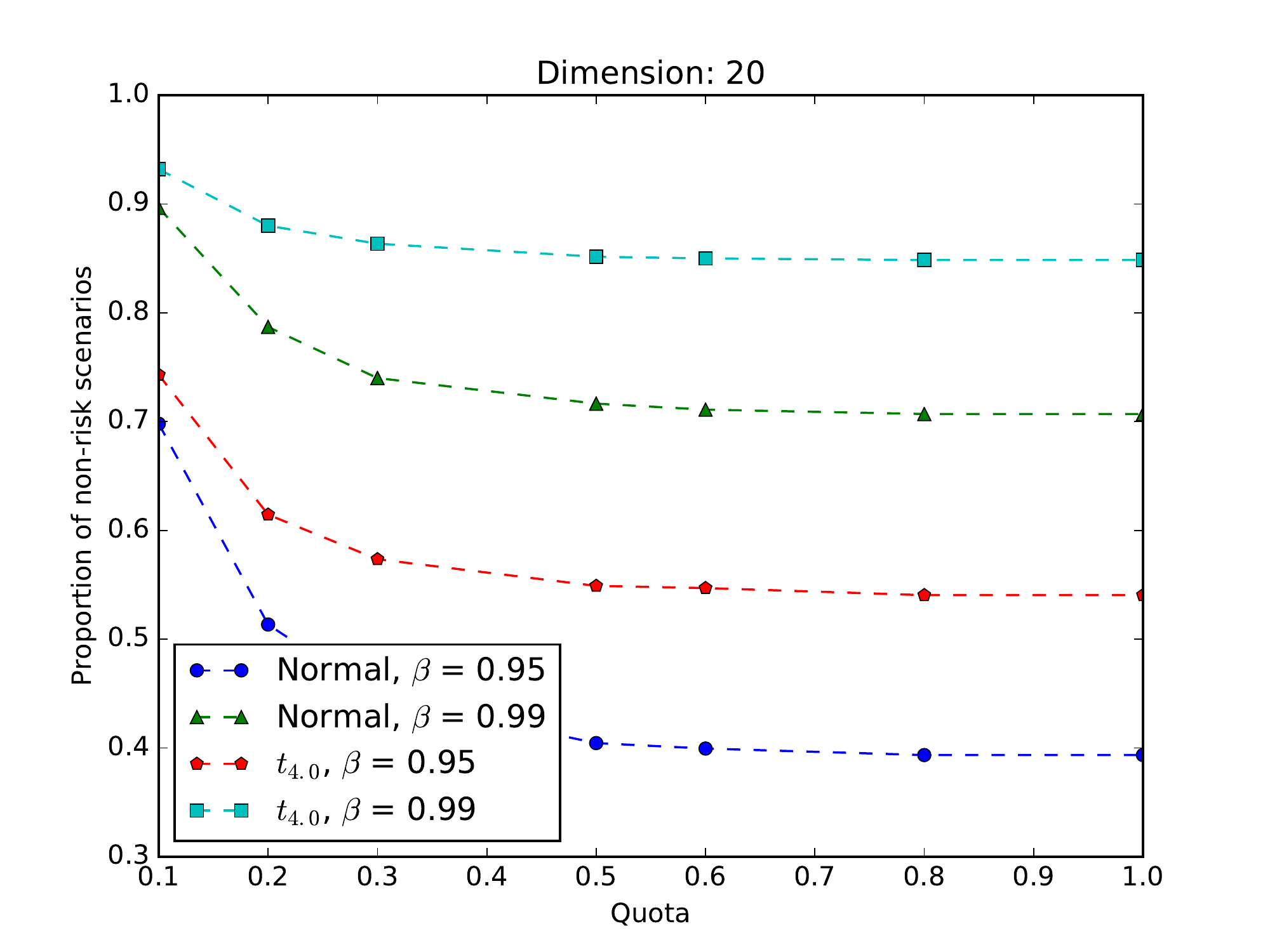}
  \end{subfigure}
  \begin{subfigure}[b]{0.45\textwidth}
    \includegraphics[width=\textwidth]{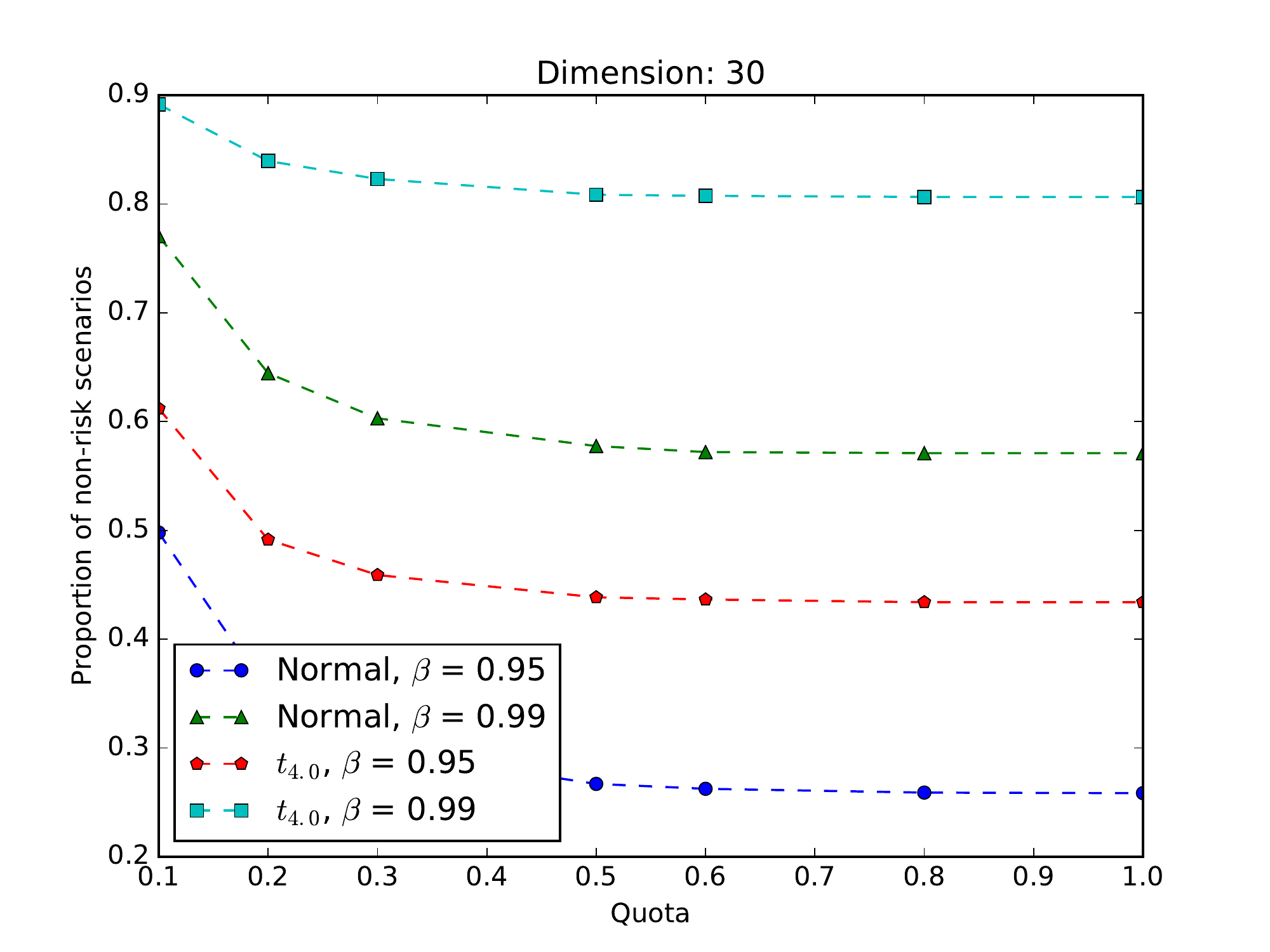}
  \end{subfigure}
  \caption{Proportions of non-risk scenarios}
  \label{fig:prop-non-risk}
\end{figure}

In Figure \ref{fig:prop-non-risk} for each each dimension we have
tested we have plotted the results of one trial. The full results can
be found in Appendix~\ref{sec:red-prop-tables}. 
The first important observation from these is that the proportion of scenarios in
the non-risk region, as compared to the uncorrelated case in Figure
\ref{fig:portfolio-cor-dim-plot}, is surprisingly high; even for $\beta=0.95$
and dimension 30, this proportion is non-negligible. 
As expected, the proportion of scenarios in the non-risk region increases as we tighten
our quota. However, for higher dimensions the quotas need
to be a lot tighter to make a significant difference.
The plots also provide further evidence that the t distribution has non-risk
regions with higher probabilities than the lighter-tailed Normal distribution.
In Figure \ref{fig:prop-non-risk}, the non-risk region for the $t$-distributions
has probability around 0.05 to 0.1 higher for dimensions 5 and 10, and
around 0.1 to 0.2 higher for dimensions 20 and 30.


\subsection{Aggregation sampling}
\label{sec:portfolio-numtests-optgap}

In this section we compare the quality of solutions
yielded by sampling and aggregation sampling by observing
the optimality gaps of the solutions that these
scenario generation methods yield. For this, we use the following
version of the portfolio selection problem.

\begin{align*}
  \minimize[x\geq 0] &\cvar\left(-x^T Y\right)\\
  \text{such that} &\ x^T\mu \geq \tau,\\
  & \sum_{i=1}^d x_i = 1,\\
  & 0 \leq x \leq u,
\end{align*}
where $\mu$ is the mean of the input distribution (rather than scenario set),
$\tau$ is the target return and $u$ is a vector
of asset quotas. The primary reason for using this formulation over those
in \eqref{eq:p2} and \eqref{eq:p3} is that given a distribution of asset
returns it is easy to select an appropriate expected target return $\tau$.
For simplicity, in our tests we set $\tau = \frac{1}{n}\sum_{i=1}^n \mu_i$ which ensures that the constraint
is feasible but not trivially satisfied.
Notice that in the above formulation we use the deterministic constraint,
$x^T \mu \geq \tau$ rather than $\E{x^TY} \geq \tau$. This is because
the latter constraint depends on the scenario set.
Therefore, the solution from a scenario-based approximation
might be infeasible with respect to the original problem, which makes measuring solution quality problematic. 

In this experiment, we test the performance of the aggregation
sampling algorithm for three families of distributions: the
Normal distribution, the t-distribution and the skew-t distribution.

For each distribution and problem dimension we run five trials
using our constructed distributions (as described in Section~\ref{sec:portfolio-experimental-set-up}).
Each trial consists of generating 50 scenario sets via sampling and aggregation
sampling, solving the corresponding scenario-based problem for each of
these sets, and calculating the optimality gap for each solution which
is yielded. For each scenario generation method we then calculate the
mean and standard deviation (S.D.) of the optimality gap. 
For the skew-t distributions, although we are able to evaluate the
objective function value for any candidate solution, to find the true
optimal solution value (or one close to it), we resort to solving the
problem for a very large sampled set of size 200000.

The full results for this experiment can be found in 
Appendix \ref{sec:agg-samp-tables}. In Figure \ref{fig:portfolio-agg-sampling}
we have plotted for one trial the raw results for dimensions 10 and 30.
We observe that there is a consistent
improvement in solution quality in using aggregation
sampling over basic sampling. In addition the solution
values are much more stable. The improvement
in solution quality and stability is particularly big for
the $t$-distributions. This is because the probability of the non-risk
region is greater for heavier-tailed distributions as observed in Section \ref{sec:portfolio-probnonrisk}.
Aggregation sampling even lead to consistently better solutions for
the skew-t distributions where we are approximating the risk region
with a risk region for a t-distribution.

\begin{figure}[h]
  \centering
  \begin{subfigure}[b]{0.45\textwidth}
    \includegraphics[width=\textwidth]{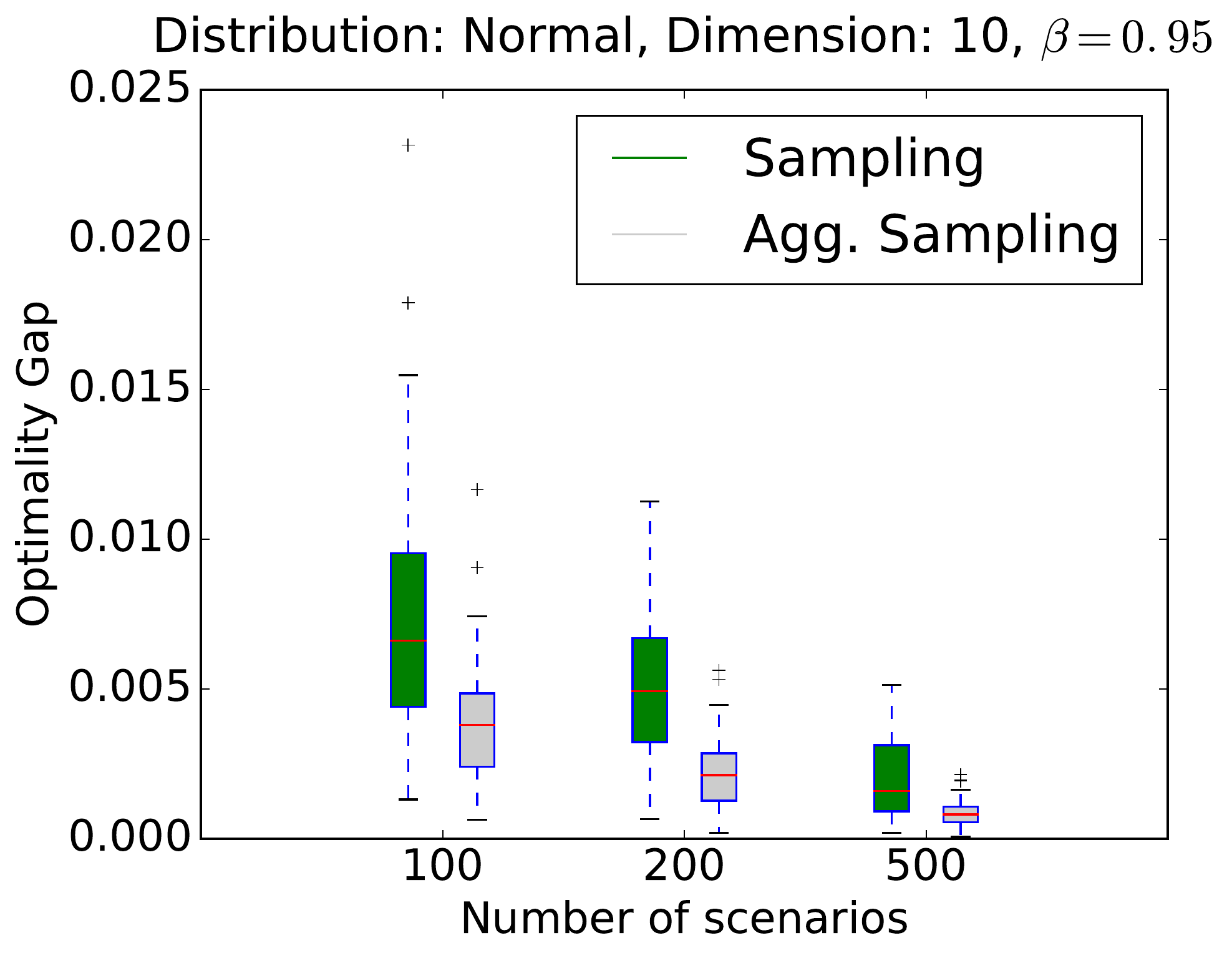}
  \end{subfigure}
  \begin{subfigure}[b]{0.45\textwidth}
    \includegraphics[width=\textwidth]{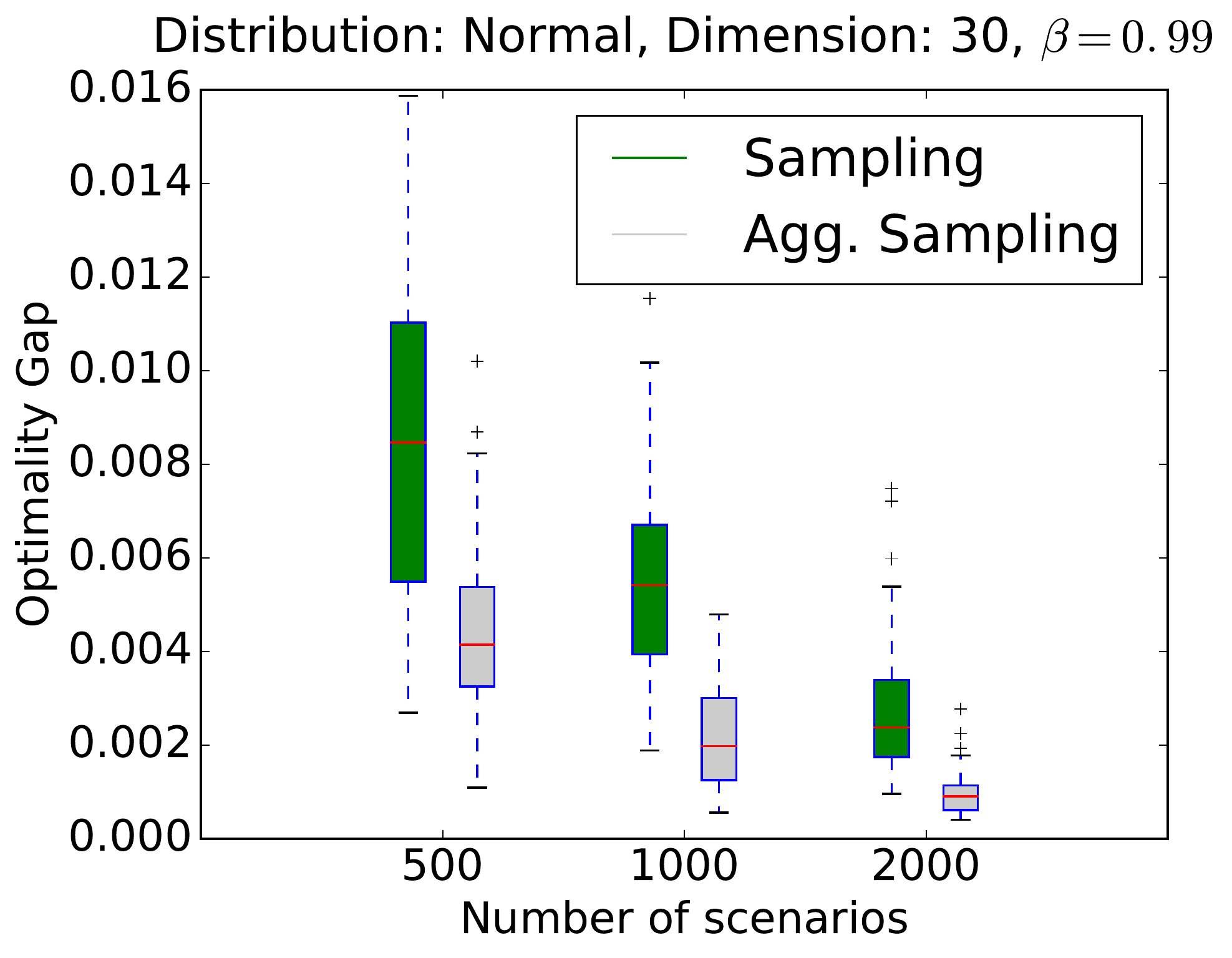}
  \end{subfigure}
  \begin{subfigure}[b]{0.45\textwidth}
    \includegraphics[width=\textwidth]{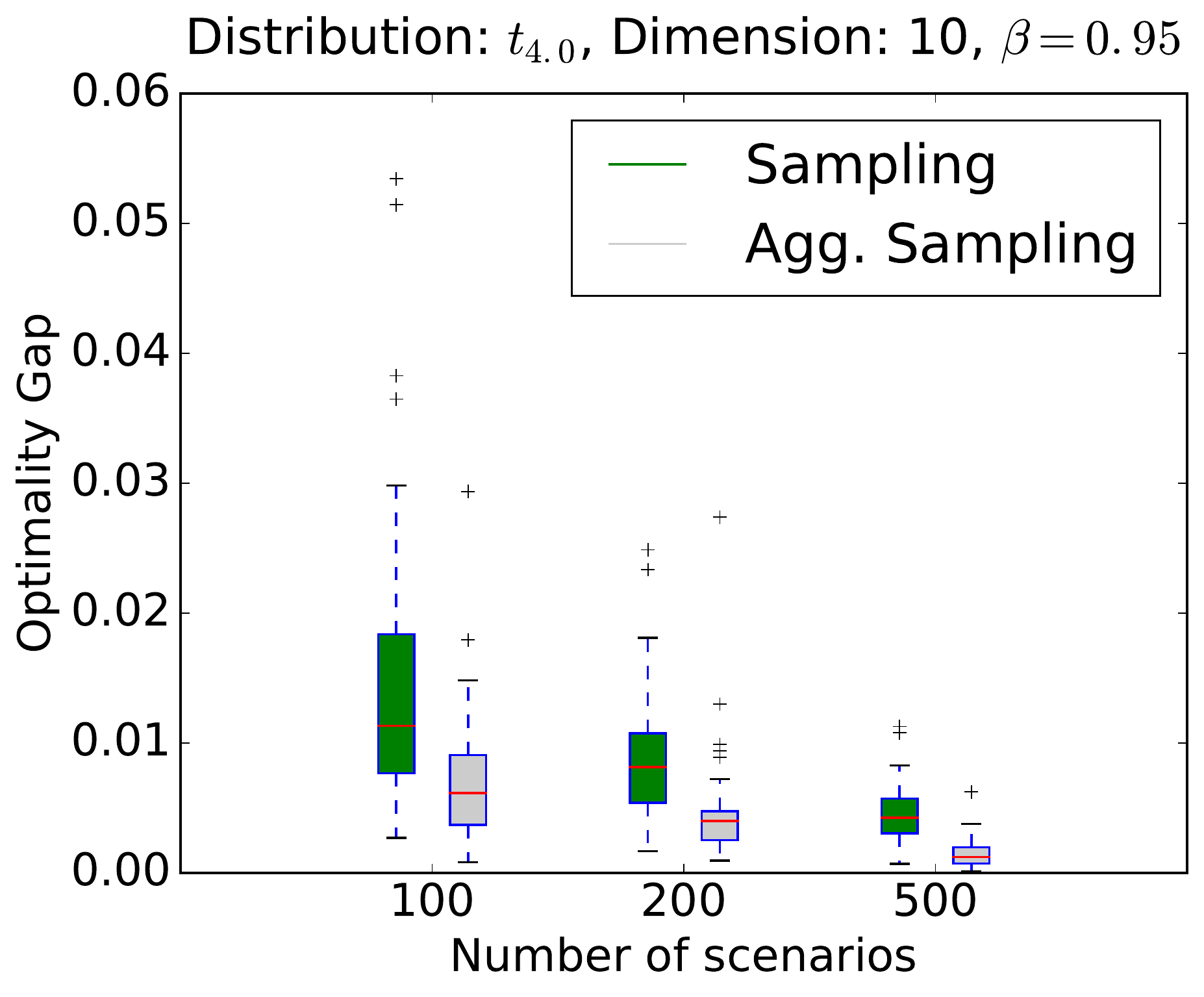}
  \end{subfigure}
  \begin{subfigure}[b]{0.45\textwidth}
    \includegraphics[width=\textwidth]{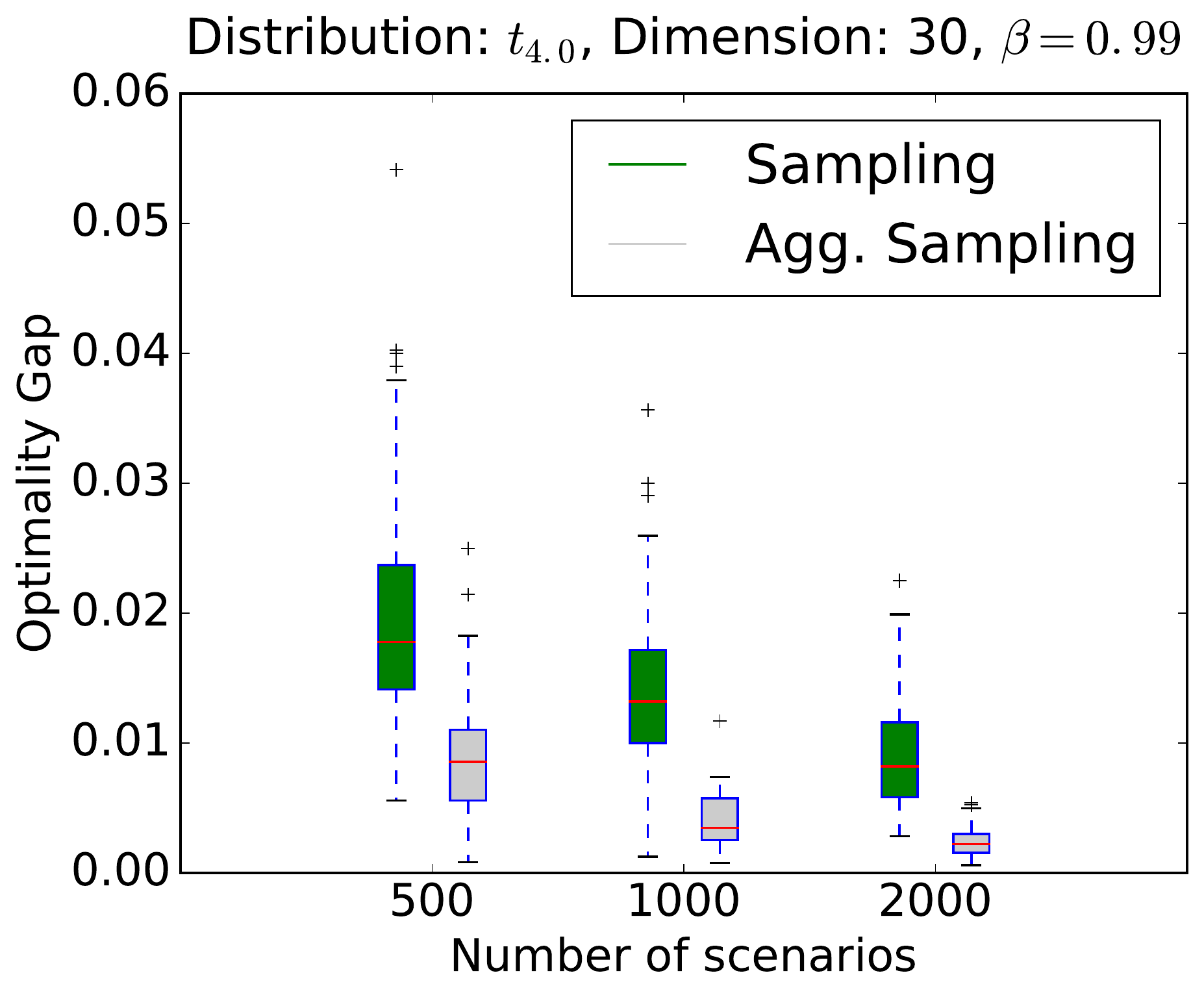}
  \end{subfigure}
  \begin{subfigure}[b]{0.45\textwidth}
    \includegraphics[width=\textwidth]{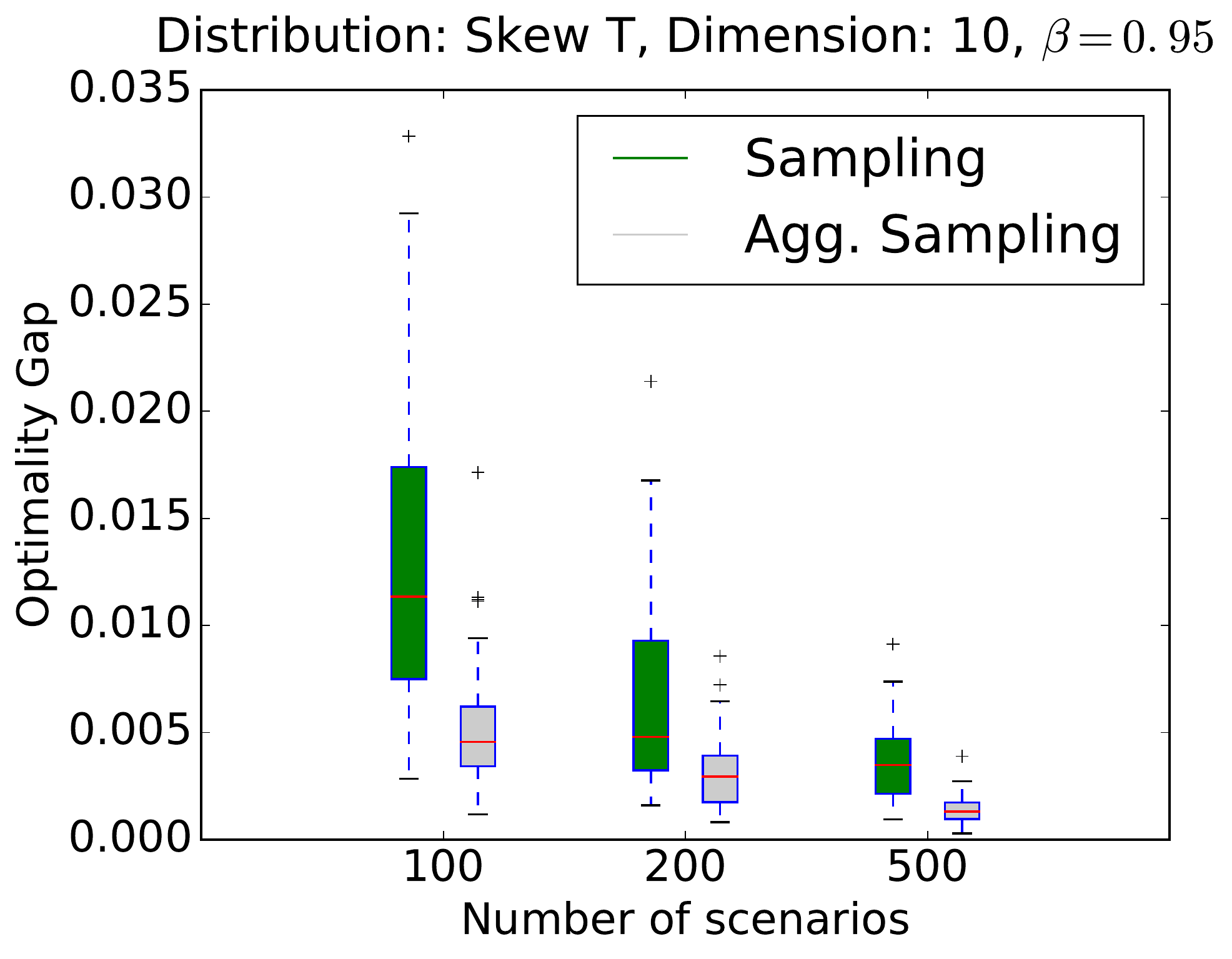}
  \end{subfigure}
  \begin{subfigure}[b]{0.45\textwidth}
    \includegraphics[width=\textwidth]{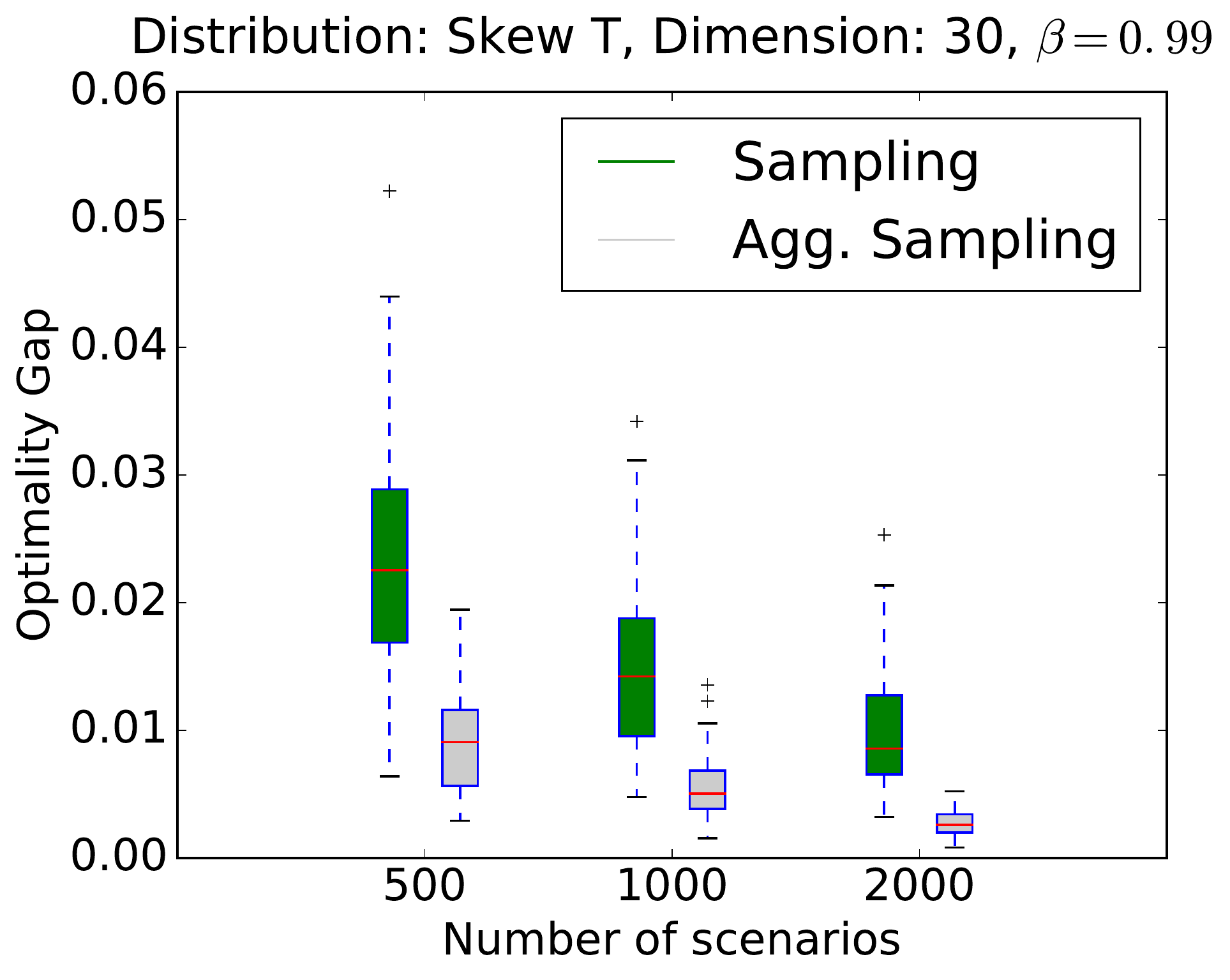}
  \end{subfigure}
  \caption{Stability test comparing performance of sampling and aggregation sampling}
  \label{fig:portfolio-agg-sampling}
\end{figure}

\subsection{Aggregation reduction}
\label{sec:portfolio-numtests-reduction}

The aim of these tests is to quantify the error induced through
the use of aggregation reduction.
In particular, we calculate the error induced in the optimal solution
value. For a given scenario set, we aggregate the non-risk scenarios, solve the problem
with respect to this reduced set, and calculate the optimality gap of this solution with respect
to the original scenario set.

For these tests we use the same problem as in Section \ref{sec:portfolio-numtests-optgap}
and run tests for Normal, t and moment matching distributions.
As explained in Section \ref{sec:portfolio-scengen}, we use the risk region
of a Normal distribution to approximate the risk region for moment-matched
scenario sets. For each family of distributions and problem dimension
we again run five trials for different instances of the distribution. 
In each trial for different initial scenario set sizes, $n=100,\ 200,\ 500$, 
and two different levels of tail risk measure $\beta=0.95,\ 0.99$,
we calculate the reduction error for 30 different scenario sets and report the mean error. 

The full results can be found in Appendix~\ref{sec:red-err-tables}.
These show that the reduction error is generally very small,
in fact for almost all problems using $\beta=0.95$, there is no error induced. 
For $\beta=0.99$, and the smallest scenario set size $n=100$,
there is a small amount of error ($< 0.01$) for the
Normal distributions, slightly larger errors for the heavier-tailed
$t$ distribution ($< 0.02$), and the largest errors (0.1-0.5) occur for reduced
moment-matching scenario sets whose risk regions have been approximated
with that of a Normal distribution. However, as the scenario set size
is increased, all errors are reduced, and for the largest scenario
set size $n=500$, there is no error induced for almost all problems.

Comparing the reduction errors with the corresponding non-risk region probabilities
in Appendix~\ref{sec:red-prop-tables}, we see that the larger errors generally
occur for the higher dimensional distributions whose non-risk
region has a larger probability. This is to be expected as the larger the non-risk
region the more scenarios that are aggregated. In Table~\ref{tab:portfolio-moments-proportions}
in Appendix~\ref{sec:red-err-tables} we have also included the proportions of reduced
scenarios for moment matching scenario sets for which we approximated the risk region
with that of a Normal distribution. The proportions of reduced scenarios in this case
are generally slightly higher than that of the corresponding Normal distributions. This
might suggest that the surrogates for the risk region are slightly too small, but this could
equally be explained by the fact that moment matching scenario sets generally
have heavier tails than the corresponding Normal distribution, which, as we observed
in Section~\ref{sec:portfolio-probnonrisk}, also leads to non-risk regions of higher probabilities.
In either case, the larger errors which are induced by reducing small moment-matching scenario
sets could be explained by these increased probabilities.

\section{Case study}
\label{sec:portfolio-case}

In this section we demonstrate the application of the presented methodology
on a difficult problem which may occur in practice.
The problem is characterized by a high-dimension, a heavy-tailed 
non-elliptical distribution of asset returns, and the presence of integer variables.
We compare the performance of the SAA method with ghost constraints
algorithm proposed in Section~\ref{sec:portfolio-scengen-ghost} against
the standard SAA method using basic sampling and aggregation
sampling without ghost constraints. 


\subsection{Problem construction}
\label{sec:problem-construction}

The following problem is used:
\begin{align*}
  \minimize[x, z] &\cvar\left(-x^T Y\right)\\
  \text{such that} &\ x^T\mu \geq \tau,\\
  & x_i \leq z_i\ \text{for each } i = 1,\ldots,d,\\
  & \sum_{i=1}^d x_i = 1,\\
  & \sum_{i=1}^d z_i = l,\\
  & 0 \leq x \leq u,\\
  & z_i \in \{0, 1\}\ \text{for each } i = 1,\ldots,d.
\end{align*}

This problem is similar to that used in Section~\ref{sec:portfolio-numtests-optgap}
except that we now use binary decision variables to limit the number of assets
in which one can invest. The extra constraints involving integer
variables may change the conic hull of feasible portfolios,
however, the method presented in Section \ref{sec:portfolio-conic-hull}
for calculating conic hulls of feasible regions cannot handle these.
We therefore ignore these constraints when constructing a risk region.
This is acceptable as the resulting conic
hull will contain the true conic hull. The problem 
is solved for $\beta=0.99$, $d=50$ assets, and maximum number of assets
$l=10$. The random vector $Y$ of asset returns is constructed by
fitting Skew-t distributions to historical monthly return data for
companies from the FTSE100 stock index.  The risk region used for this
distribution is constructed as described at the end of
Section~\ref{sec:portfolio-experimental-set-up}.

\subsection{Details of the SAA method}
\label{sec:details-saa-method}

\paragraph{Estimation of optimality gap}

In each iteration of the SAA method the optimality gap for each of the
found solutions is calculated using the procedure described in
\cite{KleywegtEA01}. Let $N$ be the scenario set size and $M$
be the number of replications used in each iteration. 
For $m = 1,\ldots,M$ denote by $g_{N}^{m}$ and
$\nu_{N}^{m}$ respectively the objective function and optimal solution
value corresponding to the $m$-th constructed scenario set. In this case,
$g_{N}^{m}(x) = \cvar(-x^{T}Y^{m})$ where $Y^{m}$ denotes the discrete random vector
corresponding to the $m$-th constructed scenario set. 

Estimators for the optimal solution value and objective function
are given by:
\begin{equation*}
  \bar{\nu}_{N}^{M} = \sum_{m=1}^{M} \nu_{N}^{m},\qquad \bar{g}_{N}^{M}(x) = \sum_{m=1}^{M} g_{N}^{M}(x),
\end{equation*}

Now, an $\alpha$-level confidence interval for the optimality
gap estimator by the solution $x$ is given by
\begin{equation*}
  \bar{g}_{N}^{M}(x) - \bar{\nu}_{N}^{M} + \Phi^{-1}(1 - \alpha)\frac{\bar{S}_{M}^{2}}{\sqrt{M}}.
\end{equation*}
where $\bar{S}_{M}^{2}$ is the standard deviation of $g_{N}^{m}(x) - \nu_{N}^{m}$
over $m=1,\ldots,M$, and $\Phi$ is the cumulative distribution
function for a standard Normal distribution.
Note that other procedures for estimating the optimality gap exist
which only require the solution of one or two problems
\cite{BayrMorton05}, \cite{stockbridge13}.

\paragraph{Sample sizes, replications and bounds}

For this experiment, the initial sample size used to solve this
problem is $N=N_{0}=200$ and at each iteration this is increased
by $N_{0}/2 = 100$. The number of replications used in each iteration
is fixed at $M=10$. These update rules have been chosen for simplicity;
more sophisticated rules for updating the sample sizes and number of
replications can be found in, for example, \cite{royset2013optimal}.

For the ghost constraints, the following heuristic rule is
used. First denote by $\hat{x}^{m}_{N}$ for $m=1,\ldots,M$
the solutions found for each replication. At the end of each iteration
the bounds $l \leq x \leq u$ are updated as follows:

\begin{align*}
  l &= \max\left(\bar{x}_{N}^{M} - \Phi^{-1}(\alpha)\frac{\bar{\sigma}^{M}_{N}}{\sqrt{M}}, 0\right),\\
  u &= \min\left(\bar{x}_{N}^{M} + \Phi^{-1}(\alpha)\frac{\bar{\sigma}^{M}_{N}}{\sqrt{M}}, 1\right),
\end{align*}
where $\bar{x}_{N}^{M}$ and $\bar{\sigma}_{N}^{M}$ are the element-wise mean
and standard deviation of the solutions $\hat{x}_{N}^{m}$ over $m=1,\ldots,M$, and the parameter $0<\alpha<1$ controls how aggressively the ghost constraints
are tightened. In this experiment we use $\alpha = 0.99$.

\paragraph{Solution validation}

Given the potential dangers in approximating the risk region, and
misspecifying ghost constraints, it is important to verify the quality of a
solution by the calculation of its corresponding
\emph{out-of-sample} value \cite{KautWall03:2}.
That is, after we calculate the $\beta$-CVaR for all candidate
solutions with respect to a large independently sampled scenario set for
all solutions yielded by the final iteration of the SAA methods. For this
experiment a sample size of 100000 is used for validation. 

\subsection{Results}
\label{sec:results}

The results to this experiment are shown in Figure~\ref{fig:case-study-res}.
In Figure~\ref{fig:case-gap} is shown the best optimality
gap found at the end of each iteration of the SAA method, in Figure~\ref{fig:case-out} are shown box plots with the out-of-sample values of the final solutions 
yielded by each method, and to aid our interpretation of the results.
In Figure~\ref{fig:case-prob} we have plotted the evolution of the
probability of the non-risk region used in the SAA method with ghost constraints.

\begin{figure}[h]
  \centering
  \begin{subfigure}[b]{0.45\textwidth}
    \includegraphics[width=\textwidth]{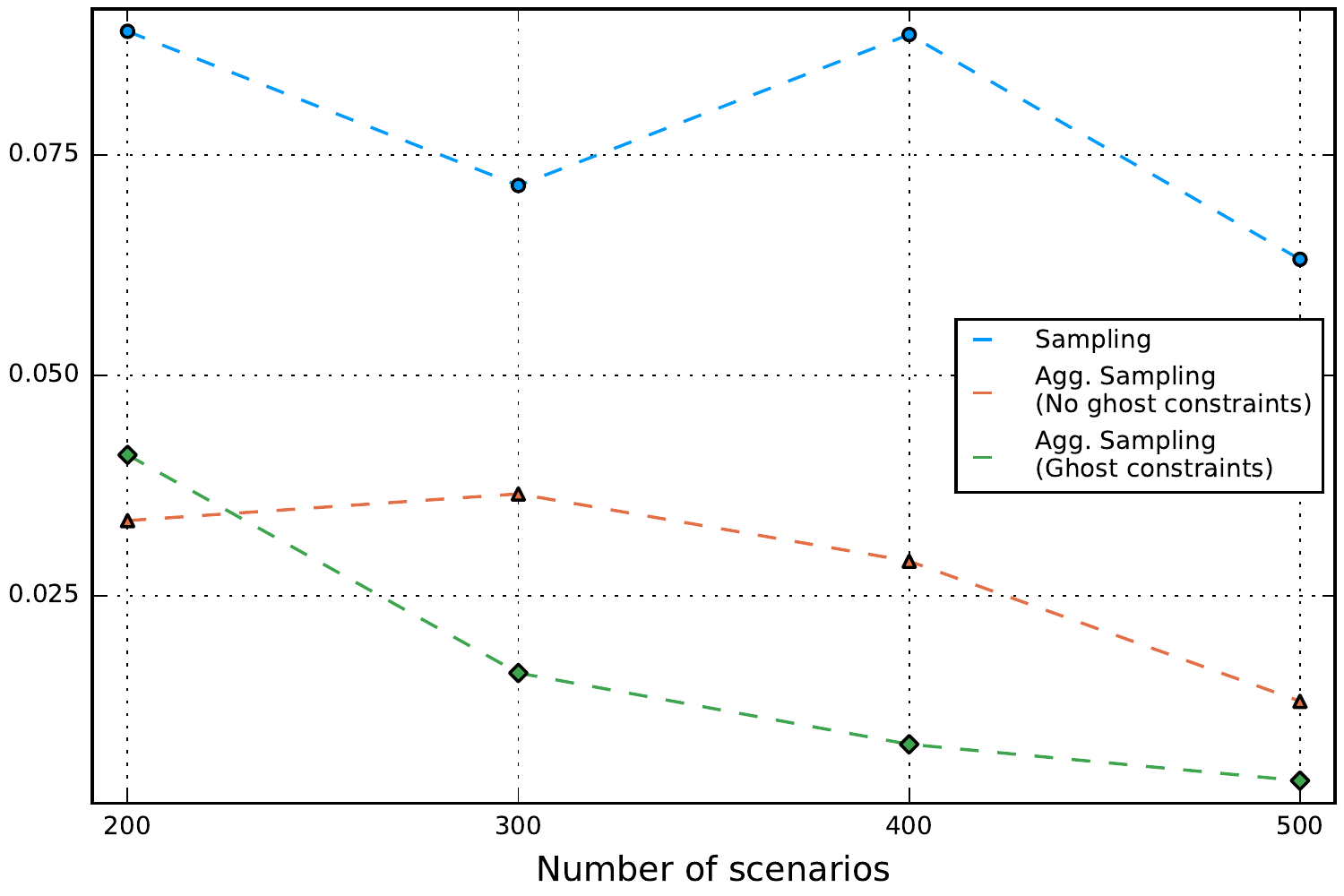}
    \caption{Best optimality gap}
    \label{fig:case-gap}
  \end{subfigure}
  \begin{subfigure}[b]{0.45\textwidth}
    \includegraphics[width=\textwidth]{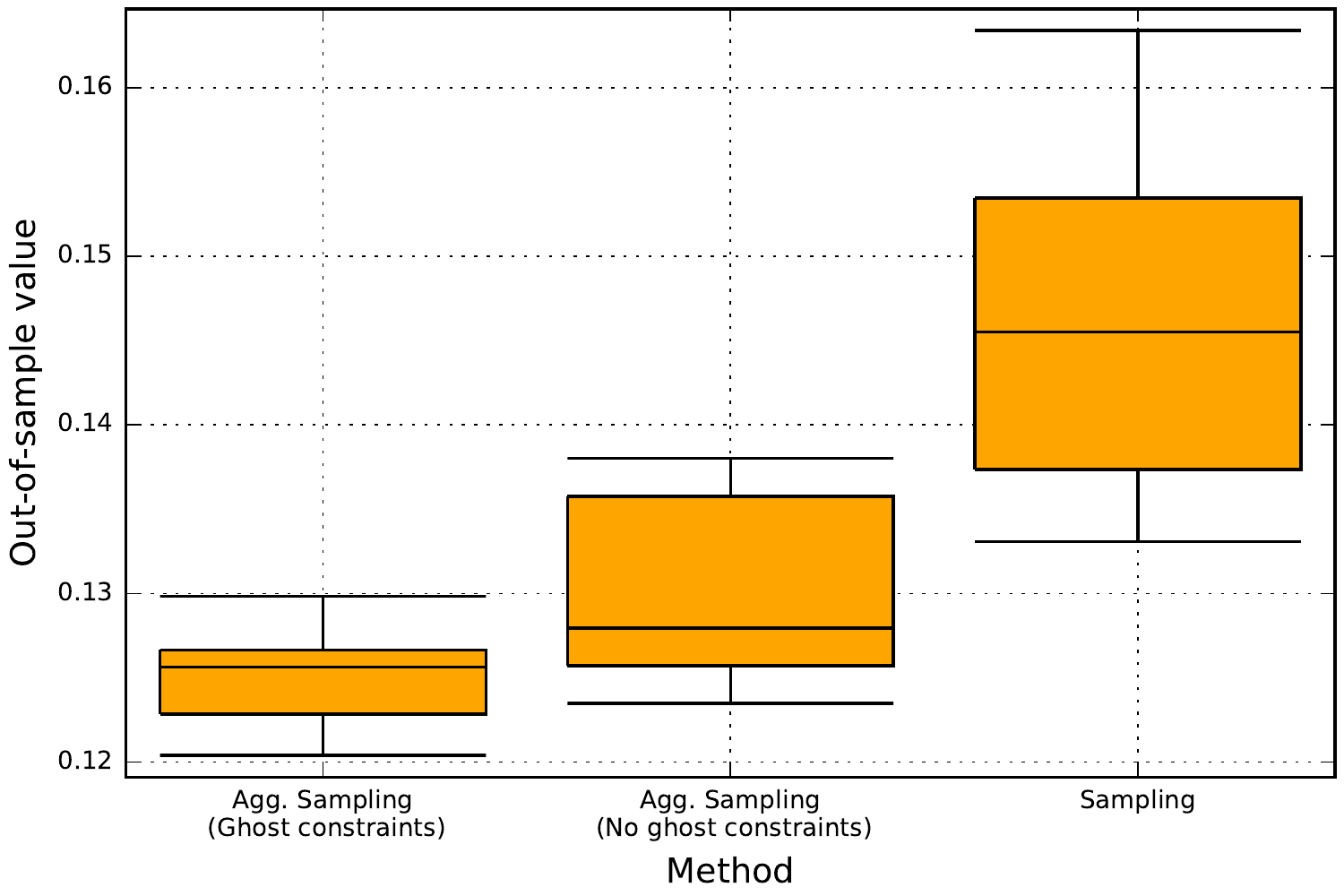}
    \caption{Out-of-sample values}
    \label{fig:case-out}
  \end{subfigure}
  \begin{subfigure}[b]{0.45\textwidth}
    \includegraphics[width=\textwidth]{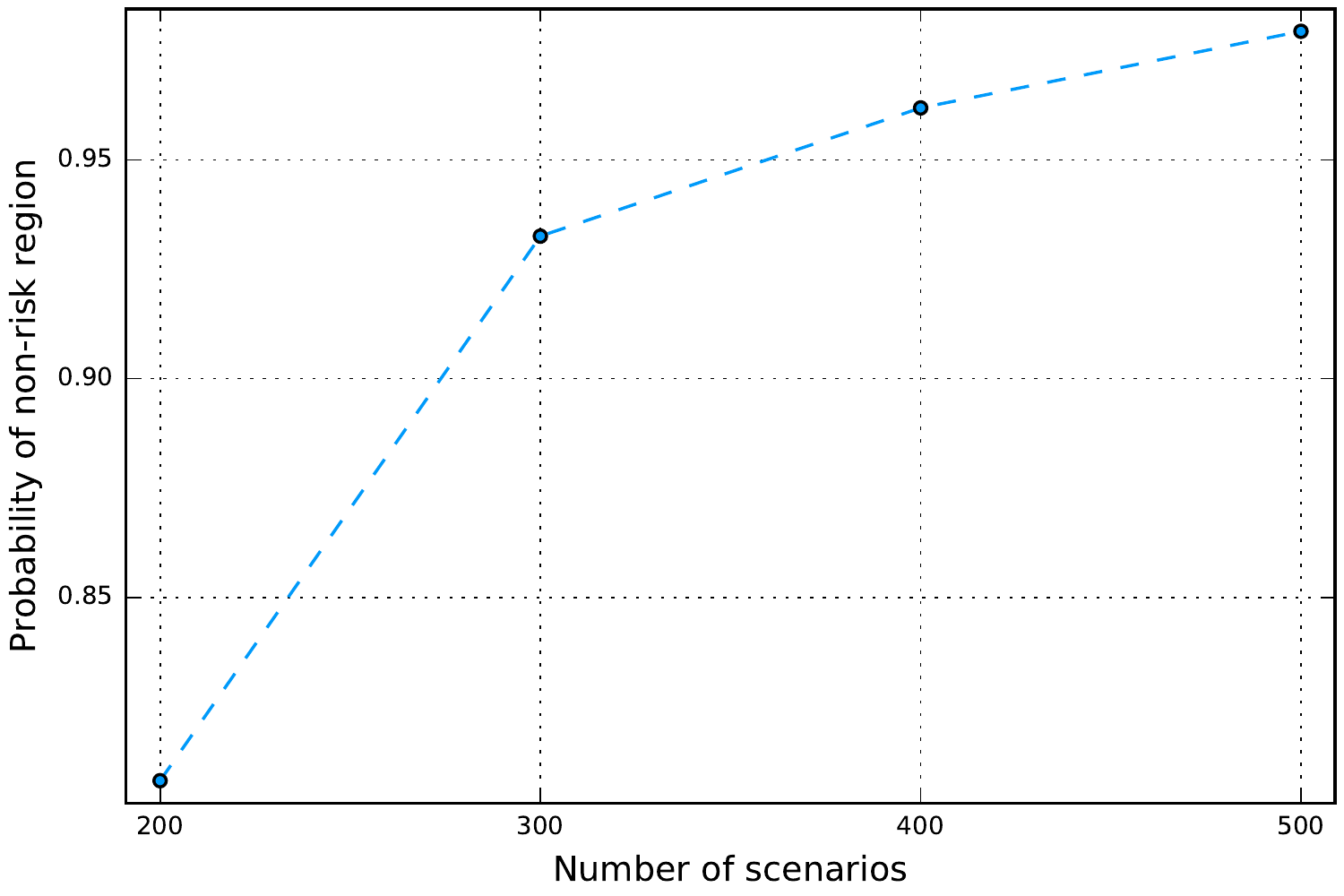}
    \caption{Prob. of non-risk region}
    \label{fig:case-prob}
  \end{subfigure}
  \caption{Case study results}
  \label{fig:case-study-res}
\end{figure}
In terms of the optimality gap and the final out-of-sample values, both
aggregation sampling methods signifiantly outperform basic sampling. For the
smallest sample size $N=200$, the best optimality gap of aggregation sampling
with and without ghost constraints is similar, since at this point no
ghost constraints have been added. As the sample size increases, the ghost
constraints become tighter and the probability of the non-risk region increases.
This leads to a much smaller best optimality gap for aggregation sampling
with ghost constraints. Since we cannot verify that the ghost constraints
added are valid, we must view this gap with some caution. However,
the out-of-sample validation after the final iteration reveals
that aggregation sampling with ghost constraints is indeed producing higher quality solutions than aggregation sampling without ghost constraints.

\section{Conclusions}
\label{sec:portfolio-conclusions}

In the paper \cite{Fairbrother15a} we proposed a general
approach to scenario generation using risk regions for stochastic programs
with tail risk measure. As proof-of-concept we demonstrated
how this applied for portfolio selection problems for elliptically distributed returns. In this work, we have presented how this methodology may be used
for more realistic portfolio selection problems, in particular those which are high-dimensional, have non-elliptical assets returns and integer decision variables.


The main issue in applying the methodology to more realistic problems
was its extension to non-elliptical distributions. In order to do
this, we proposed the use of approximate risk regions, and derived
results which indicated our approach would be robust against small
misspecifications of risk region. Although this paper was focused
on portfolio selection, the results were derived for general stochastic
programs, which means that using approximate risk regions 
could work for other problems with tail risk measures.

We tested the performance of our methodology
for solving realistic problems where the return
distributions were fitted from real financial return
data. Aggregation sampling generally outperformed 
basic sampling in terms of solution quality and 
stability. We also showed that aggregation reduction
induces almost no error in the solution for reasonably sized scenario
sets. These results not only held for elliptical
distributions, but also non-elliptical distributions
for which we used approximate risk regions.

The effectiveness of using risk regions for scenario
generation depends upon the probability of the risk region:
the greater the probability of the non-risk region,
the more scenarios that can be aggregated. 
It follows directly from the definition of risk
regions that this probability decreases as the problem
becomes more constrained. Based on this observation,
a heuristic based on the SAA method was proposed, which
adds artificial constraints, called ghost constraints,
to the problem. As the algorithm progresses, the ghost
constraints are tightened which in effect allows
one to focus in on high quality solutions. This
algorithm was demonstrated on a difficult case
study problem, and was shown to significantly outperform
basic sampling and aggregation sampling without ghost constraints.
This algorithm was presented in a non-problem specific way and could
potentially be applied to other stochastic programs with tail risk
measure.




\bibliographystyle{alpha}
\bibliography{bibtex-files/stein}

\appendix

\newgeometry{left=1cm,bottom=2.0cm, right=0cm}

\begin{landscape}
  \section{Reduction proportion tables}
  \label{sec:red-prop-tables}
  The following tables list the estimated probabilities of the non-risk region
  for a variety of distributions constructed from real data. See Section~\ref{sec:portfolio-numtests-prob}
  for details. Each table corresponds to a family of
  distributions at a given dimension, and each row gives the proportions
  for a given set of companies. In addition, the distributions corresponding the i-th row of each table 
  of dimension $d$ have been fitted using the same set of companies.
  \begin{table}[h!]
    \centering
    \begin{tabular}{|c|c|c|c|c|c|c|c|c|c|c|c|c|c|}
\hline
\multicolumn{7}{|c|}{$\beta = 0.95$} & \multicolumn{7}{|c|}{$\beta = 0.99$} \\
\hline
$x \leq 1.0$ & $x \leq 0.8$ & $x \leq 0.6$ & $x \leq 0.5$ & $x \leq 0.4$ & $x \leq 0.3$ & $x \leq 0.2$ & $x \leq 1.0$ & $x \leq 0.8$ & $x \leq 0.6$ & $x \leq 0.5$ & $x \leq 0.4$ & $x \leq 0.3$ & $x \leq 0.2$ \\
\hline
0.743 & 0.752 & 0.767 & 0.782 & 0.814 & 0.865 & 0.950 & 0.934 & 0.936 & 0.941 & 0.946 & 0.959 & 0.971 & 0.995 \\
0.738 & 0.744 & 0.760 & 0.777 & 0.809 & 0.855 & 0.949 & 0.922 & 0.925 & 0.928 & 0.934 & 0.949 & 0.965 & 0.992 \\
0.767 & 0.775 & 0.793 & 0.807 & 0.832 & 0.872 & 0.948 & 0.930 & 0.932 & 0.937 & 0.943 & 0.953 & 0.969 & 0.990 \\
0.763 & 0.771 & 0.784 & 0.801 & 0.830 & 0.880 & 0.951 & 0.931 & 0.934 & 0.944 & 0.949 & 0.957 & 0.973 & 0.987 \\
0.755 & 0.763 & 0.777 & 0.798 & 0.829 & 0.883 & 0.955 & 0.927 & 0.929 & 0.935 & 0.940 & 0.951 & 0.966 & 0.991 \\
\hline
\end{tabular}
\caption{Proportion of reduced scenarios for Normal distributed returns and $d = 5$}

  \end{table}
  \vspace{-0.6cm}
  \begin{table}[h!]
    \centering
    \begin{tabular}{|c|c|c|c|c|c|c|c|c|c|c|c|c|c|}
\hline
\multicolumn{7}{|c|}{$\beta = 0.95$} & \multicolumn{7}{|c|}{$\beta = 0.99$} \\
\hline
$x \leq 1.0$ & $x \leq 0.8$ & $x \leq 0.6$ & $x \leq 0.5$ & $x \leq 0.4$ & $x \leq 0.3$ & $x \leq 0.2$ & $x \leq 1.0$ & $x \leq 0.8$ & $x \leq 0.6$ & $x \leq 0.5$ & $x \leq 0.4$ & $x \leq 0.3$ & $x \leq 0.2$ \\
\hline
0.594 & 0.600 & 0.608 & 0.617 & 0.637 & 0.679 & 0.752 & 0.833 & 0.834 & 0.839 & 0.846 & 0.860 & 0.882 & 0.917 \\
0.617 & 0.621 & 0.632 & 0.647 & 0.669 & 0.703 & 0.777 & 0.851 & 0.852 & 0.856 & 0.860 & 0.868 & 0.879 & 0.914 \\
0.506 & 0.509 & 0.523 & 0.534 & 0.560 & 0.606 & 0.689 & 0.779 & 0.780 & 0.787 & 0.791 & 0.806 & 0.837 & 0.889 \\
0.564 & 0.566 & 0.573 & 0.590 & 0.615 & 0.658 & 0.748 & 0.827 & 0.828 & 0.835 & 0.846 & 0.857 & 0.877 & 0.921 \\
0.537 & 0.540 & 0.552 & 0.566 & 0.586 & 0.624 & 0.727 & 0.820 & 0.822 & 0.825 & 0.832 & 0.843 & 0.870 & 0.912 \\
\hline
\end{tabular}
\caption{Proportion of reduced scenarios for Normal distributed returns and $d = 10$}

  \end{table}
  \vspace{-0.6cm}
  \begin{table}[h!]
    \centering
    \begin{tabular}{|c|c|c|c|c|c|c|c|c|c|c|c|c|c|}
\hline
\multicolumn{7}{|c|}{$\beta = 0.95$} & \multicolumn{7}{|c|}{$\beta = 0.99$} \\
\hline
$x \leq 1.0$ & $x \leq 0.8$ & $x \leq 0.6$ & $x \leq 0.5$ & $x \leq 0.3$ & $x \leq 0.2$ & $x \leq 0.1$ & $x \leq 1.0$ & $x \leq 0.8$ & $x \leq 0.6$ & $x \leq 0.5$ & $x \leq 0.3$ & $x \leq 0.2$ & $x \leq 0.1$ \\
\hline
0.394 & 0.394 & 0.400 & 0.405 & 0.447 & 0.513 & 0.698 & 0.707 & 0.707 & 0.711 & 0.717 & 0.740 & 0.787 & 0.896 \\
0.325 & 0.326 & 0.332 & 0.342 & 0.392 & 0.457 & 0.635 & 0.653 & 0.653 & 0.655 & 0.662 & 0.696 & 0.740 & 0.851 \\
0.344 & 0.344 & 0.348 & 0.354 & 0.389 & 0.460 & 0.668 & 0.648 & 0.648 & 0.653 & 0.656 & 0.683 & 0.743 & 0.870 \\
0.384 & 0.385 & 0.390 & 0.401 & 0.440 & 0.507 & 0.708 & 0.695 & 0.695 & 0.698 & 0.704 & 0.740 & 0.782 & 0.896 \\
0.417 & 0.418 & 0.424 & 0.432 & 0.479 & 0.540 & 0.738 & 0.727 & 0.727 & 0.730 & 0.735 & 0.764 & 0.813 & 0.906 \\
\hline
\end{tabular}
\caption{Proportion of reduced scenarios for Normal distributed returns and $d = 20$}

  \end{table}
  \vspace{-0.6cm}
  \begin{table}[h!]
    \centering
    \begin{tabular}{|c|c|c|c|c|c|c|c|c|c|c|c|c|c|}
\hline
\multicolumn{7}{|c|}{$\beta = 0.95$} & \multicolumn{7}{|c|}{$\beta = 0.99$} \\
\hline
$x \leq 1.0$ & $x \leq 0.8$ & $x \leq 0.6$ & $x \leq 0.5$ & $x \leq 0.3$ & $x \leq 0.2$ & $x \leq 0.1$ & $x \leq 1.0$ & $x \leq 0.8$ & $x \leq 0.6$ & $x \leq 0.5$ & $x \leq 0.3$ & $x \leq 0.2$ & $x \leq 0.1$ \\
\hline
0.259 & 0.259 & 0.263 & 0.267 & 0.297 & 0.350 & 0.498 & 0.571 & 0.571 & 0.572 & 0.578 & 0.603 & 0.644 & 0.770 \\
0.264 & 0.266 & 0.269 & 0.272 & 0.299 & 0.347 & 0.511 & 0.587 & 0.587 & 0.589 & 0.591 & 0.616 & 0.661 & 0.790 \\
0.282 & 0.282 & 0.286 & 0.291 & 0.321 & 0.378 & 0.533 & 0.599 & 0.599 & 0.602 & 0.607 & 0.631 & 0.681 & 0.785 \\
0.247 & 0.247 & 0.251 & 0.257 & 0.281 & 0.333 & 0.502 & 0.555 & 0.555 & 0.556 & 0.558 & 0.586 & 0.630 & 0.769 \\
0.293 & 0.293 & 0.296 & 0.301 & 0.324 & 0.374 & 0.548 & 0.583 & 0.583 & 0.584 & 0.587 & 0.613 & 0.665 & 0.802 \\
\hline
\end{tabular}
\caption{Proportion of reduced scenarios for Normal distributed returns and $d = 30$}

  \end{table}
  \vspace{-0.6cm}
  \begin{table}[h!]
    \centering
    \begin{tabular}{|c|c|c|c|c|c|c|c|c|c|c|c|c|c|}
\hline
\multicolumn{7}{|c|}{$\beta = 0.95$} & \multicolumn{7}{|c|}{$\beta = 0.99$} \\
\hline
$x \leq 1.0$ & $x \leq 0.8$ & $x \leq 0.6$ & $x \leq 0.5$ & $x \leq 0.4$ & $x \leq 0.3$ & $x \leq 0.2$ & $x \leq 1.0$ & $x \leq 0.8$ & $x \leq 0.6$ & $x \leq 0.5$ & $x \leq 0.4$ & $x \leq 0.3$ & $x \leq 0.2$ \\
\hline
0.793 & 0.801 & 0.814 & 0.822 & 0.842 & 0.876 & 0.950 & 0.952 & 0.953 & 0.957 & 0.960 & 0.966 & 0.976 & 0.992 \\
0.775 & 0.782 & 0.796 & 0.812 & 0.837 & 0.877 & 0.946 & 0.949 & 0.950 & 0.954 & 0.956 & 0.961 & 0.972 & 0.988 \\
0.808 & 0.815 & 0.829 & 0.841 & 0.859 & 0.898 & 0.953 & 0.958 & 0.960 & 0.962 & 0.964 & 0.969 & 0.980 & 0.992 \\
0.799 & 0.808 & 0.819 & 0.828 & 0.855 & 0.882 & 0.950 & 0.949 & 0.951 & 0.954 & 0.957 & 0.965 & 0.977 & 0.990 \\
0.793 & 0.799 & 0.809 & 0.822 & 0.848 & 0.887 & 0.951 & 0.960 & 0.960 & 0.963 & 0.965 & 0.969 & 0.976 & 0.991 \\
\hline
\end{tabular}
\caption{Proportion of reduced scenarios for $t_{4.0}$ distributed returns and $d = 5$}

  \end{table}
  \vspace{-0.6cm}
  \begin{table}[h!]
    \centering
    \begin{tabular}{|c|c|c|c|c|c|c|c|c|c|c|c|c|c|}
\hline
\multicolumn{7}{|c|}{$\beta = 0.95$} & \multicolumn{7}{|c|}{$\beta = 0.99$} \\
\hline
$x \leq 1.0$ & $x \leq 0.8$ & $x \leq 0.6$ & $x \leq 0.5$ & $x \leq 0.4$ & $x \leq 0.3$ & $x \leq 0.2$ & $x \leq 1.0$ & $x \leq 0.8$ & $x \leq 0.6$ & $x \leq 0.5$ & $x \leq 0.4$ & $x \leq 0.3$ & $x \leq 0.2$ \\
\hline
0.689 & 0.691 & 0.700 & 0.709 & 0.720 & 0.750 & 0.804 & 0.916 & 0.916 & 0.917 & 0.919 & 0.926 & 0.931 & 0.949 \\
0.711 & 0.713 & 0.719 & 0.730 & 0.742 & 0.769 & 0.829 & 0.923 & 0.924 & 0.925 & 0.926 & 0.930 & 0.940 & 0.956 \\
0.616 & 0.617 & 0.630 & 0.640 & 0.656 & 0.677 & 0.754 & 0.895 & 0.896 & 0.898 & 0.900 & 0.905 & 0.915 & 0.935 \\
0.642 & 0.642 & 0.647 & 0.657 & 0.672 & 0.703 & 0.783 & 0.896 & 0.896 & 0.900 & 0.904 & 0.913 & 0.925 & 0.941 \\
0.652 & 0.655 & 0.666 & 0.675 & 0.690 & 0.723 & 0.785 & 0.905 & 0.905 & 0.907 & 0.908 & 0.913 & 0.924 & 0.944 \\
\hline
\end{tabular}
\caption{Proportion of reduced scenarios for $t_{4.0}$ distributed returns and $d = 10$}

  \end{table}
  \vspace{-0.6cm}
  \begin{table}[h!]
    \centering
    \begin{tabular}{|c|c|c|c|c|c|c|c|c|c|c|c|c|c|}
\hline
\multicolumn{7}{|c|}{$\beta = 0.95$} & \multicolumn{7}{|c|}{$\beta = 0.99$} \\
\hline
$x \leq 1.0$ & $x \leq 0.8$ & $x \leq 0.6$ & $x \leq 0.5$ & $x \leq 0.3$ & $x \leq 0.2$ & $x \leq 0.1$ & $x \leq 1.0$ & $x \leq 0.8$ & $x \leq 0.6$ & $x \leq 0.5$ & $x \leq 0.3$ & $x \leq 0.2$ & $x \leq 0.1$ \\
\hline
0.540 & 0.540 & 0.547 & 0.549 & 0.574 & 0.615 & 0.743 & 0.849 & 0.849 & 0.850 & 0.852 & 0.864 & 0.880 & 0.932 \\
0.461 & 0.463 & 0.467 & 0.475 & 0.509 & 0.560 & 0.703 & 0.835 & 0.836 & 0.840 & 0.844 & 0.858 & 0.870 & 0.919 \\
0.506 & 0.507 & 0.510 & 0.515 & 0.551 & 0.595 & 0.753 & 0.839 & 0.839 & 0.839 & 0.840 & 0.855 & 0.874 & 0.931 \\
0.511 & 0.511 & 0.514 & 0.519 & 0.562 & 0.612 & 0.753 & 0.860 & 0.860 & 0.862 & 0.865 & 0.876 & 0.894 & 0.939 \\
0.567 & 0.568 & 0.572 & 0.576 & 0.609 & 0.657 & 0.797 & 0.866 & 0.867 & 0.867 & 0.870 & 0.881 & 0.901 & 0.952 \\
\hline
\end{tabular}
\caption{Proportion of reduced scenarios for $t_{4.0}$ distributed returns and $d = 20$}

  \end{table}
  \vspace{-0.6cm}
  \begin{table}[h!]
    \centering
    \begin{tabular}{|c|c|c|c|c|c|c|c|c|c|c|c|c|c|}
\hline
\multicolumn{7}{|c|}{$\beta = 0.95$} & \multicolumn{7}{|c|}{$\beta = 0.99$} \\
\hline
$x \leq 1.0$ & $x \leq 0.8$ & $x \leq 0.6$ & $x \leq 0.5$ & $x \leq 0.3$ & $x \leq 0.2$ & $x \leq 0.1$ & $x \leq 1.0$ & $x \leq 0.8$ & $x \leq 0.6$ & $x \leq 0.5$ & $x \leq 0.3$ & $x \leq 0.2$ & $x \leq 0.1$ \\
\hline
0.434 & 0.434 & 0.436 & 0.439 & 0.459 & 0.491 & 0.612 & 0.806 & 0.806 & 0.807 & 0.808 & 0.823 & 0.840 & 0.891 \\
0.466 & 0.466 & 0.468 & 0.469 & 0.495 & 0.532 & 0.649 & 0.821 & 0.821 & 0.823 & 0.824 & 0.838 & 0.853 & 0.897 \\
0.443 & 0.443 & 0.445 & 0.448 & 0.474 & 0.512 & 0.637 & 0.821 & 0.822 & 0.822 & 0.824 & 0.834 & 0.854 & 0.898 \\
0.444 & 0.445 & 0.448 & 0.454 & 0.470 & 0.513 & 0.635 & 0.812 & 0.813 & 0.814 & 0.814 & 0.823 & 0.841 & 0.889 \\
0.417 & 0.417 & 0.419 & 0.421 & 0.444 & 0.487 & 0.617 & 0.808 & 0.808 & 0.810 & 0.811 & 0.823 & 0.844 & 0.891 \\
\hline
\end{tabular}
\caption{Proportion of reduced scenarios for $t_{4.0}$ distributed returns and $d = 30$}

  \end{table}
  \vspace{-0.6cm}
\end{landscape}

\newgeometry{left=1cm,bottom=2.0cm, right=2cm}
\section{Aggregation sampling tables}
\label{sec:agg-samp-tables}

The following tables list the relative reduction in the mean and standard deviation of optimality gaps
for aggregation sampling compared with sampling for a variety of distributions. See 
Section~\ref{sec:portfolio-numtests-optgap} for more details.

\begin{table}[h!]
  \centering
  \begin{tabular}{|c|c|c|c|c|c|}
\hline
\multicolumn{2}{|c|}{n = 100} & \multicolumn{2}{|c|}{n = 200} & \multicolumn{2}{|c|}{n = 500} \\
\hline
Mean Imp. & S.D. Imp. & Mean Imp. & S.D. Imp. & Mean Imp. & S.D. Imp. \\
\hline
2.747 & 2.542 & 3.226 & 3.321 & 3.697 & 2.871 \\
3.905 & 4.427 & 3.226 & 3.323 & 3.646 & 4.439 \\
3.803 & 2.993 & 4.889 & 3.538 & 4.567 & 3.927 \\
3.376 & 3.040 & 3.402 & 2.517 & 5.182 & 4.357 \\
3.240 & 3.257 & 3.432 & 2.246 & 4.807 & 4.708 \\
\hline
\end{tabular}
\caption{ Comparison for $d = 5$, $\beta = 0.95$, and Normal returns }

  \label{tab:portfolio-agg-sampling-normal_5}
\end{table}

\begin{table}[h!]
  \centering
  \begin{tabular}{|c|c|c|c|c|c|}
\hline
\multicolumn{2}{|c|}{n = 100} & \multicolumn{2}{|c|}{n = 200} & \multicolumn{2}{|c|}{n = 500} \\
\hline
Mean Imp. & S.D. Imp. & Mean Imp. & S.D. Imp. & Mean Imp. & S.D. Imp. \\
\hline
1.989 & 1.876 & 2.670 & 2.422 & 2.460 & 2.495 \\
2.018 & 2.494 & 2.711 & 2.227 & 3.126 & 2.864 \\
1.559 & 1.652 & 1.736 & 1.230 & 2.727 & 2.678 \\
1.869 & 2.089 & 2.275 & 2.181 & 2.551 & 2.731 \\
1.996 & 2.085 & 2.285 & 2.061 & 2.466 & 2.828 \\
\hline
\end{tabular}
\caption{ Comparison for $d = 10$, $\beta = 0.95$, and Normal returns }

  \label{tab:portfolio-agg-sampling-normal_10}
\end{table}

\begin{table}[h!]
  \centering
  \begin{tabular}{|c|c|c|c|c|c|}
\hline
\multicolumn{2}{|c|}{n = 500} & \multicolumn{2}{|c|}{n = 1000} & \multicolumn{2}{|c|}{n = 2000} \\
\hline
Mean Imp. & S.D. Imp. & Mean Imp. & S.D. Imp. & Mean Imp. & S.D. Imp. \\
\hline
2.357 & 2.124 & 2.890 & 3.039 & 3.026 & 2.809 \\
2.504 & 3.054 & 2.750 & 2.839 & 2.873 & 2.689 \\
2.308 & 1.963 & 2.546 & 2.854 & 2.803 & 2.791 \\
2.341 & 2.699 & 2.948 & 3.369 & 2.592 & 2.367 \\
2.802 & 2.657 & 3.421 & 2.494 & 3.725 & 3.547 \\
\hline
\end{tabular}
\caption{ Comparison for $d = 20$, $\beta = 0.99$, and Normal returns }

  \label{tab:portfolio-agg-sampling-normal_20}
\end{table}

\begin{table}[h!]
  \centering
  \begin{tabular}{|c|c|c|c|c|c|}
\hline
\multicolumn{2}{|c|}{n = 500} & \multicolumn{2}{|c|}{n = 1000} & \multicolumn{2}{|c|}{n = 2000} \\
\hline
Mean Imp. & S.D. Imp. & Mean Imp. & S.D. Imp. & Mean Imp. & S.D. Imp. \\
\hline
1.943 & 1.842 & 2.161 & 2.148 & 2.901 & 2.846 \\
1.779 & 2.195 & 2.197 & 2.067 & 2.590 & 2.483 \\
1.990 & 2.227 & 2.246 & 2.033 & 2.405 & 2.514 \\
2.019 & 2.012 & 2.076 & 2.057 & 2.010 & 1.891 \\
1.866 & 1.769 & 2.457 & 1.921 & 2.853 & 3.138 \\
\hline
\end{tabular}
\caption{ Comparison for $d = 30$, $\beta = 0.99$, and Normal returns }

  \label{tab:portfolio-agg-sampling-normal_30}
\end{table}

\begin{table}[h!]
  \centering
  \begin{tabular}{|c|c|c|c|c|c|}
\hline
\multicolumn{2}{|c|}{n = 100} & \multicolumn{2}{|c|}{n = 200} & \multicolumn{2}{|c|}{n = 500} \\
\hline
Mean Imp. & S.D. Imp. & Mean Imp. & S.D. Imp. & Mean Imp. & S.D. Imp. \\
\hline
2.857 & 2.661 & 2.762 & 1.981 & 3.500 & 3.709 \\
3.407 & 3.431 & 3.692 & 3.416 & 5.572 & 6.167 \\
4.335 & 3.062 & 3.872 & 4.195 & 3.244 & 3.149 \\
4.280 & 3.748 & 4.636 & 6.732 & 4.974 & 6.593 \\
2.578 & 1.773 & 3.664 & 3.500 & 4.019 & 4.160 \\
\hline
\end{tabular}
\caption{ Comparison for $d = 5$, $\beta = 0.95$, and $t_{4.0}$ returns }

  \label{tab:portfolio-agg-sampling-tdist_5}
\end{table}

\begin{table}[h!]
  \centering
  \begin{tabular}{|c|c|c|c|c|c|}
\hline
\multicolumn{2}{|c|}{n = 100} & \multicolumn{2}{|c|}{n = 200} & \multicolumn{2}{|c|}{n = 500} \\
\hline
Mean Imp. & S.D. Imp. & Mean Imp. & S.D. Imp. & Mean Imp. & S.D. Imp. \\
\hline
1.899 & 2.091 & 2.169 & 1.805 & 2.939 & 2.599 \\
2.078 & 1.910 & 2.358 & 2.229 & 2.982 & 2.340 \\
1.996 & 2.923 & 2.639 & 3.126 & 2.088 & 1.727 \\
2.658 & 2.958 & 2.436 & 2.222 & 2.357 & 2.312 \\
2.080 & 2.171 & 1.980 & 1.232 & 2.957 & 2.114 \\
\hline
\end{tabular}
\caption{ Comparison for $d = 10$, $\beta = 0.95$, and $t_{4.0}$ returns }

  \label{tab:portfolio-agg-sampling-tdist_10}
\end{table}

\begin{table}[h!]
  \centering
  \begin{tabular}{|c|c|c|c|c|c|}
\hline
\multicolumn{2}{|c|}{n = 500} & \multicolumn{2}{|c|}{n = 1000} & \multicolumn{2}{|c|}{n = 2000} \\
\hline
Mean Imp. & S.D. Imp. & Mean Imp. & S.D. Imp. & Mean Imp. & S.D. Imp. \\
\hline
4.142 & 5.028 & 4.215 & 4.383 & 5.571 & 5.221 \\
3.039 & 3.843 & 4.096 & 4.346 & 4.857 & 6.084 \\
3.378 & 3.831 & 4.020 & 4.267 & 5.007 & 5.617 \\
3.722 & 4.886 & 3.744 & 3.247 & 4.339 & 5.336 \\
3.616 & 3.524 & 4.999 & 3.739 & 5.116 & 6.277 \\
\hline
\end{tabular}
\caption{ Comparison for $d = 20$, $\beta = 0.99$, and $t_{4.0}$ returns }

  \label{tab:portfolio-agg-sampling-tdist_20}
\end{table}

\begin{table}[h!]
  \centering
  \begin{tabular}{|c|c|c|c|c|c|}
\hline
\multicolumn{2}{|c|}{n = 500} & \multicolumn{2}{|c|}{n = 1000} & \multicolumn{2}{|c|}{n = 2000} \\
\hline
Mean Imp. & S.D. Imp. & Mean Imp. & S.D. Imp. & Mean Imp. & S.D. Imp. \\
\hline
3.035 & 3.068 & 2.950 & 2.547 & 3.741 & 4.042 \\
2.359 & 1.983 & 3.513 & 5.068 & 3.384 & 3.029 \\
3.507 & 4.356 & 2.977 & 3.966 & 3.686 & 4.915 \\
2.950 & 3.005 & 3.079 & 1.964 & 3.936 & 4.240 \\
2.228 & 2.043 & 3.549 & 3.227 & 3.950 & 4.267 \\
\hline
\end{tabular}
\caption{ Comparison for $d = 30$, $\beta = 0.99$, and $t_{4.0}$ returns }

  \label{tab:portfolio-agg-sampling-tdist_30}
\end{table}

\begin{table}[h!]
  \centering
  \begin{tabular}{|c|c|c|c|c|c|}
\hline
\multicolumn{2}{|c|}{n = 100} & \multicolumn{2}{|c|}{n = 200} & \multicolumn{2}{|c|}{n = 500} \\
\hline
Mean Imp. & S.D. Imp. & Mean Imp. & S.D. Imp. & Mean Imp. & S.D. Imp. \\
\hline
1.917 & 1.601 & 2.766 & 3.020 & 3.352 & 2.644 \\
1.887 & 1.857 & 2.748 & 2.416 & 3.414 & 3.290 \\
3.171 & 3.489 & 4.433 & 3.427 & 3.949 & 3.774 \\
2.620 & 3.170 & 3.038 & 3.518 & 2.872 & 3.178 \\
2.391 & 2.408 & 2.027 & 1.891 & 3.466 & 3.434 \\
\hline
\end{tabular}
\caption{ Comparison for $d = 5$, $\beta = 0.95$, and Skew T returns }

  \label{tab:portfolio-agg-sampling-skewtdist_5}
\end{table}

\begin{table}[h!]
  \centering
  \begin{tabular}{|c|c|c|c|c|c|}
\hline
\multicolumn{2}{|c|}{n = 100} & \multicolumn{2}{|c|}{n = 200} & \multicolumn{2}{|c|}{n = 500} \\
\hline
Mean Imp. & S.D. Imp. & Mean Imp. & S.D. Imp. & Mean Imp. & S.D. Imp. \\
\hline
1.839 & 2.189 & 2.215 & 1.925 & 2.977 & 2.650 \\
1.631 & 2.021 & 2.203 & 2.087 & 2.150 & 2.554 \\
1.962 & 1.671 & 1.872 & 1.187 & 3.172 & 3.513 \\
1.627 & 1.868 & 1.661 & 2.136 & 1.775 & 1.439 \\
2.502 & 2.417 & 2.152 & 2.577 & 2.647 & 2.580 \\
\hline
\end{tabular}
\caption{ Comparison for $d = 10$, $\beta = 0.95$, and Skew T returns }

  \label{tab:portfolio-agg-sampling-skewtdist_10}
\end{table}

\begin{table}[h!]
  \centering
  \begin{tabular}{|c|c|c|c|c|c|}
\hline
\multicolumn{2}{|c|}{n = 500} & \multicolumn{2}{|c|}{n = 1000} & \multicolumn{2}{|c|}{n = 2000} \\
\hline
Mean Imp. & S.D. Imp. & Mean Imp. & S.D. Imp. & Mean Imp. & S.D. Imp. \\
\hline
4.646 & 5.803 & 4.921 & 4.384 & 5.843 & 6.268 \\
4.639 & 4.025 & 6.296 & 5.028 & 6.513 & 7.438 \\
3.355 & 3.840 & 3.655 & 3.163 & 3.305 & 3.359 \\
3.317 & 2.257 & 3.448 & 3.623 & 4.794 & 4.732 \\
3.395 & 3.365 & 3.164 & 3.145 & 4.351 & 4.306 \\
\hline
\end{tabular}
\caption{ Comparison for $d = 20$, $\beta = 0.99$, and Skew T returns }

  \label{tab:portfolio-agg-sampling-skewtdist_20}
\end{table}

\begin{table}[h!]
  \centering
  \begin{tabular}{|c|c|c|c|c|c|}
\hline
\multicolumn{2}{|c|}{n = 500} & \multicolumn{2}{|c|}{n = 1000} & \multicolumn{2}{|c|}{n = 2000} \\
\hline
Mean Imp. & S.D. Imp. & Mean Imp. & S.D. Imp. & Mean Imp. & S.D. Imp. \\
\hline
2.631 & 3.659 & 3.364 & 4.298 & 4.000 & 4.099 \\
2.285 & 2.809 & 2.667 & 3.201 & 3.482 & 2.882 \\
3.266 & 4.545 & 3.617 & 4.340 & 3.791 & 3.138 \\
2.923 & 3.334 & 3.750 & 3.796 & 4.304 & 5.492 \\
2.486 & 2.289 & 2.658 & 2.754 & 3.659 & 4.918 \\
\hline
\end{tabular}
\caption{ Comparison for $d = 30$, $\beta = 0.99$, and Skew T returns }

  \label{tab:portfolio-agg-sampling-skewtdist_30}
\end{table}

\newgeometry{left=1cm,bottom=2.0cm, right=0cm}
\begin{landscape}
\section{Reduction error tables}
\label{sec:red-err-tables}

The following tables list the mean error induced
by aggregating scenarios in the non-risk region for a variety of distributions. 
See Section~\ref{sec:portfolio-numtests-reduction} for details.

\centering
\vspace{0.5cm}
\begin{minipage}[b]{0.5\textwidth}
  \begin{tabular}{|c|c|c|c|c|c|}
\hline
\multicolumn{3}{|c|}{$\beta = 0.95$} & \multicolumn{3}{|c|}{$\beta = 0.99$} \\
\hline
$n = 100$ & $n = 200$ & $n = 500$ & $n = 100$ & $n = 200$ & $n = 500$ \\
\hline
0.000 & 0.000 & 0.000 & 0.008 & 0.002 & 0.000 \\
0.000 & 0.000 & 0.000 & 0.009 & 0.001 & 0.000 \\
0.000 & 0.000 & 0.000 & 0.005 & 0.002 & 0.000 \\
0.000 & 0.000 & 0.000 & 0.007 & 0.001 & 0.000 \\
0.000 & 0.000 & -0.000 & 0.007 & 0.001 & 0.000 \\
\hline
\end{tabular}
\captionof{table}{ Reduction error induced for d=5 Normal returns }

\end{minipage}
\hspace{0.8cm}
\begin{minipage}[b]{0.5\textwidth}
  \begin{tabular}{|c|c|c|c|c|c|}
\hline
\multicolumn{3}{|c|}{$\beta = 0.95$} & \multicolumn{3}{|c|}{$\beta = 0.99$} \\
\hline
$n = 100$ & $n = 200$ & $n = 500$ & $n = 100$ & $n = 200$ & $n = 500$ \\
\hline
0.000 & 0.000 & 0.000 & 0.006 & 0.000 & 0.000 \\
0.000 & 0.000 & 0.000 & 0.006 & 0.001 & 0.000 \\
0.000 & -0.000 & 0.000 & 0.006 & 0.001 & 0.000 \\
0.000 & 0.000 & -0.000 & 0.004 & 0.000 & 0.000 \\
0.000 & -0.000 & -0.000 & 0.005 & 0.000 & 0.000 \\
\hline
\end{tabular}
\captionof{table}{ Reduction error induced for d=10 Normal returns }

\end{minipage}

\vspace{0.2cm}

\begin{minipage}[b]{0.5\textwidth}
  \begin{tabular}{|c|c|c|c|c|c|}
\hline
\multicolumn{3}{|c|}{$\beta = 0.95$} & \multicolumn{3}{|c|}{$\beta = 0.99$} \\
\hline
$n = 100$ & $n = 200$ & $n = 500$ & $n = 100$ & $n = 200$ & $n = 500$ \\
\hline
0.000 & 0.000 & -0.000 & 0.003 & 0.000 & 0.000 \\
0.000 & -0.000 & -0.000 & 0.002 & 0.000 & 0.000 \\
0.000 & -0.000 & -0.000 & 0.002 & 0.000 & 0.000 \\
0.000 & 0.000 & 0.000 & 0.002 & 0.000 & -0.000 \\
0.000 & 0.000 & 0.000 & 0.003 & 0.000 & 0.000 \\
\hline
\end{tabular}
\captionof{table}{ Reduction error induced for d=20 Normal returns }

\end{minipage}
\hspace{0.8cm}
\begin{minipage}[b]{0.5\textwidth}
  \begin{tabular}{|c|c|c|c|c|c|}
\hline
\multicolumn{3}{|c|}{$\beta = 0.95$} & \multicolumn{3}{|c|}{$\beta = 0.99$} \\
\hline
$n = 100$ & $n = 200$ & $n = 500$ & $n = 100$ & $n = 200$ & $n = 500$ \\
\hline
0.000 & -0.000 & 0.000 & 0.001 & 0.000 & 0.000 \\
0.000 & -0.000 & 0.000 & 0.002 & 0.000 & 0.000 \\
-0.000 & 0.000 & -0.000 & 0.002 & 0.000 & -0.000 \\
-0.000 & -0.000 & 0.000 & 0.001 & 0.000 & -0.000 \\
0.000 & 0.000 & 0.000 & 0.001 & 0.000 & 0.000 \\
\hline
\end{tabular}
\captionof{table}{ Reduction error induced for d=30 Normal returns }

\end{minipage}

\vspace{0.2cm}

\begin{minipage}[b]{0.5\textwidth}
  \begin{tabular}{|c|c|c|c|c|c|}
\hline
\multicolumn{3}{|c|}{$\beta = 0.95$} & \multicolumn{3}{|c|}{$\beta = 0.99$} \\
\hline
$n = 100$ & $n = 200$ & $n = 500$ & $n = 100$ & $n = 200$ & $n = 500$ \\
\hline
0.001 & 0.000 & 0.000 & 0.014 & 0.005 & 0.000 \\
0.000 & 0.000 & 0.000 & 0.012 & 0.003 & 0.000 \\
0.000 & 0.000 & 0.000 & 0.017 & 0.002 & 0.001 \\
0.000 & 0.000 & 0.000 & 0.015 & 0.005 & 0.000 \\
0.001 & 0.000 & 0.000 & 0.015 & 0.004 & 0.001 \\
\hline
\end{tabular}
\captionof{table}{ Reduction error induced for d=5 $t_{4.0}$ returns }

\end{minipage}
\hspace{0.8cm}
\begin{minipage}[b]{0.5\textwidth}
  \begin{tabular}{|c|c|c|c|c|c|}
\hline
\multicolumn{3}{|c|}{$\beta = 0.95$} & \multicolumn{3}{|c|}{$\beta = 0.99$} \\
\hline
$n = 100$ & $n = 200$ & $n = 500$ & $n = 100$ & $n = 200$ & $n = 500$ \\
\hline
0.000 & 0.000 & 0.000 & 0.015 & 0.003 & 0.000 \\
0.000 & 0.000 & 0.000 & 0.015 & 0.004 & 0.000 \\
0.000 & -0.000 & 0.000 & 0.012 & 0.002 & -0.000 \\
0.000 & 0.000 & -0.000 & 0.017 & 0.004 & 0.000 \\
0.000 & -0.000 & -0.000 & 0.020 & 0.003 & 0.000 \\
\hline
\end{tabular}
\captionof{table}{ Reduction error induced for d=10 $t_{4.0}$ returns }

\end{minipage}

\vspace{0.2cm}

\begin{minipage}[b]{0.5\textwidth}
  \begin{tabular}{|c|c|c|c|c|c|}
\hline
\multicolumn{3}{|c|}{$\beta = 0.95$} & \multicolumn{3}{|c|}{$\beta = 0.99$} \\
\hline
$n = 100$ & $n = 200$ & $n = 500$ & $n = 100$ & $n = 200$ & $n = 500$ \\
\hline
0.000 & 0.000 & -0.000 & 0.013 & 0.003 & 0.000 \\
0.000 & 0.000 & -0.000 & 0.017 & 0.000 & 0.000 \\
0.000 & 0.000 & -0.000 & 0.016 & 0.003 & 0.000 \\
0.000 & 0.000 & 0.000 & 0.012 & 0.002 & 0.000 \\
0.000 & -0.000 & 0.000 & 0.016 & 0.002 & 0.000 \\
\hline
\end{tabular}
\captionof{table}{ Reduction error induced for d=20 $t_{4.0}$ returns }

\end{minipage}
\hspace{0.8cm}
\begin{minipage}[b]{0.5\textwidth}
  \begin{tabular}{|c|c|c|c|c|c|}
\hline
\multicolumn{3}{|c|}{$\beta = 0.95$} & \multicolumn{3}{|c|}{$\beta = 0.99$} \\
\hline
$n = 100$ & $n = 200$ & $n = 500$ & $n = 100$ & $n = 200$ & $n = 500$ \\
\hline
0.000 & -0.000 & 0.000 & 0.015 & 0.001 & 0.000 \\
0.000 & 0.000 & 0.000 & 0.015 & 0.004 & 0.000 \\
0.000 & -0.000 & 0.000 & 0.013 & 0.002 & 0.000 \\
0.000 & -0.000 & -0.000 & 0.015 & 0.004 & 0.000 \\
0.000 & 0.000 & 0.000 & 0.016 & 0.004 & 0.000 \\
\hline
\end{tabular}
\captionof{table}{ Reduction error induced for d=30 $t_{4.0}$ returns }

\end{minipage}

\begin{minipage}[b]{0.5\textwidth}
  \begin{tabular}{|c|c|c|c|c|c|}
\hline
\multicolumn{3}{|c|}{$\beta = 0.95$} & \multicolumn{3}{|c|}{$\beta = 0.99$} \\
\hline
$n = 100$ & $n = 200$ & $n = 500$ & $n = 100$ & $n = 200$ & $n = 500$ \\
\hline
0.000 & -0.000 & 0.000 & 0.001 & 0.000 & -0.000 \\
0.000 & 0.000 & 0.000 & 0.002 & 0.000 & 0.000 \\
0.000 & 0.000 & 0.000 & 0.000 & -0.000 & 0.000 \\
-0.000 & 0.000 & 0.000 & 0.000 & 0.000 & -0.000 \\
0.000 & 0.000 & 0.000 & 0.001 & 0.001 & 0.000 \\
\hline
\end{tabular}
\captionof{table}{ Reduction error induced for d=5 Moment Matching returns }

\end{minipage}
\hspace{0.8cm}
\begin{minipage}[b]{0.5\textwidth}
  \begin{tabular}{|c|c|c|c|c|c|}
\hline
\multicolumn{3}{|c|}{$\beta = 0.95$} & \multicolumn{3}{|c|}{$\beta = 0.99$} \\
\hline
$n = 100$ & $n = 200$ & $n = 500$ & $n = 100$ & $n = 200$ & $n = 500$ \\
\hline
-0.000 & -0.000 & -0.000 & 0.003 & 0.000 & 0.000 \\
-0.000 & -0.000 & 0.000 & 0.001 & -0.000 & -0.000 \\
0.000 & 0.000 & 0.000 & 0.509 & 0.001 & 0.001 \\
-0.000 & 0.000 & -0.000 & 0.003 & 0.000 & 0.000 \\
0.000 & 0.000 & 0.000 & 0.002 & 0.000 & 0.000 \\
\hline
\end{tabular}
\captionof{table}{ Reduction error induced for d=10 Moment Matching returns }

\end{minipage}

\vspace{0.2cm}

\begin{minipage}[b]{0.5\textwidth}
  \begin{tabular}{|c|c|c|c|c|c|}
\hline
\multicolumn{3}{|c|}{$\beta = 0.95$} & \multicolumn{3}{|c|}{$\beta = 0.99$} \\
\hline
$n = 100$ & $n = 200$ & $n = 500$ & $n = 100$ & $n = 200$ & $n = 500$ \\
\hline
0.000 & 0.000 & 0.000 & 0.089 & 0.000 & 0.000 \\
-0.000 & 0.000 & 0.000 & 0.407 & 0.003 & 0.000 \\
0.000 & -0.000 & -0.000 & 0.003 & 0.001 & 0.000 \\
0.000 & 0.000 & 0.000 & 0.231 & 0.001 & 0.000 \\
0.000 & 0.000 & 0.000 & 0.000 & 0.000 & 0.000 \\
\hline
\end{tabular}
\captionof{table}{ Reduction error induced for d=20 Moment Matching returns }

\end{minipage}
\hspace{0.8cm}
\begin{minipage}[b]{0.5\textwidth}
  \begin{tabular}{|c|c|c|c|c|c|}
\hline
\multicolumn{3}{|c|}{$\beta = 0.95$} & \multicolumn{3}{|c|}{$\beta = 0.99$} \\
\hline
$n = 100$ & $n = 200$ & $n = 500$ & $n = 100$ & $n = 200$ & $n = 500$ \\
\hline
0.000 & 0.000 & 0.000 & 0.205 & 0.000 & 0.000 \\
0.000 & 0.000 & 0.000 & 0.111 & 0.001 & 0.000 \\
0.000 & 0.000 & 0.000 & 0.206 & 0.001 & 0.000 \\
0.000 & 0.000 & 0.000 & 0.218 & 0.001 & 0.000 \\
0.000 & 0.000 & 0.000 & 0.071 & 0.001 & 0.000 \\
\hline
\end{tabular}
\captionof{table}{ Reduction error induced for d=30 Moment Matching returns }

\end{minipage}

\begin{table}[h!]
  \centering
  \begin{tabular}{|c|c|c|c|c|c|c|c|}
\hline
\multicolumn{2}{|c|}{$d=5$} & \multicolumn{2}{|c|}{$d=10$} & \multicolumn{2}{|c|}{$d=20$} & \multicolumn{2}{|c|}{$d=30$} \\
\hline
$\beta=0.95$ & $\beta=0.99$ & $\beta=0.95$ & $\beta=0.99$ & $\beta=0.95$ & $\beta=0.99$ & $\beta=0.95$ & $\beta=0.99$ \\
\hline
0.786 & 0.919 & 0.638 & 0.840 & 0.481 & 0.734 & 0.380 & 0.646 \\
0.743 & 0.900 & 0.623 & 0.827 & 0.477 & 0.741 & 0.365 & 0.647 \\
0.761 & 0.905 & 0.660 & 0.869 & 0.445 & 0.729 & 0.381 & 0.650 \\
0.770 & 0.907 & 0.625 & 0.847 & 0.455 & 0.716 & 0.366 & 0.655 \\
0.747 & 0.917 & 0.640 & 0.860 & 0.446 & 0.712 & 0.333 & 0.616 \\
\hline
\end{tabular}

  \caption{Proportions of scenarios reduced for moment matching scenario sets}
  \label{tab:portfolio-moments-proportions}
\end{table}

\end{landscape}


\end{document}